\documentclass[a4paper,UKenglish,cleveref, autoref, thm-restate,authorcolumns]{lipics-v2019}

\usepackage{microtype}
\usepackage{amsthm,amsmath}
\usepackage{thmtools}
\usepackage{amssymb}
\usepackage{hyperref}
\usepackage{graphicx,inputenc}
\usepackage[normalem]{ulem}

\usepackage[utf8]{inputenc}

\usepackage{lipsum}       
\usepackage{changepage}   

\usepackage{xspace,cite,dsfont}

\usepackage{xcolor}

\usepackage[ruled, vlined,boxed,commentsnumbered,linesnumbered]{algorithm2e}
\definecolor{dark-red}{rgb}{0.4,0.15,0.15}
\definecolor{dark-blue}{rgb}{0.15,0.15,0.4}
\definecolor{medium-blue}{rgb}{0,0,0.5}
\definecolor{gray}{rgb}{0.5,0.5,0.5}
\definecolor{color-Ig}{rgb}{0.15,0.7,0.15}

\hypersetup{
    colorlinks, linkcolor={dark-red},
    citecolor={dark-blue}, urlcolor={medium-blue}
}

\newcommand{\NP}{\ensuremath{\mathsf{NP}}\xspace}
\newcommand{\BPP}{\ensuremath{\mathsf{BPP}}\xspace}
\renewcommand{\P}{\ensuremath{\mathsf{P}}\xspace}

\newcommand{\Acal}{\mathcal{A}}

\newcommand{\Ocal}{\mathcal{O}}

\newcommand{\Rcal}{\mathcal{R}}

\newcommand{\FPT}{{\sf FPT}\xspace}
\renewcommand{\deg}{{\sf deg}\xspace}
\newcommand{\dec}{{\sf dec}\xspace}

\newcommand{\fixme}[1]{}
\newcommand{\fixmeperso}[1]{}
\newcommand{\lop}{{\sf lop}\xspace}
\newcommand{\opt}{{\sf opt}_{\Pi}\xspace}
\renewcommand{\P}{{\sf P}\xspace}
\newcommand{\val}{{\sf val}_{\Pi}\xspace}
\newcommand{\ub}{{\sf u}\xspace}

\newcommand{\probl}[3]{
  \begin{flushleft}
    \fbox{
      \begin{minipage}{.98\textwidth}
        \noindent {\sc #1}\\
        {\bf Input:} #2\\
        {\bf Question:} #3
      \end{minipage}}
  \end{flushleft}
}

\renewcommand{\O}{\mathcal{O}}

\newcommand{\yes}{{\sf yes}\xspace}
\newcommand{\no}{{\sf no}\xspace}

\newcommand{\MSOone}{{\sf MSO}$_1$\xspace}

\newcommand{\PPT}{{\sf PPT}\xspace}

\newcommand{\todo}[1][]{%
  \ifx/#1/%
    \textcolor{red}{TODO!}%
  \else%
    \textcolor{red}{todo: #1}%
  \fi%
}
\newcommand{\Oh}{\ensuremath{\mathcal{O}}\xspace}

\theoremstyle{plain}

\newtheorem{observation}[theorem]{Observation}
\newtheorem{conjecture}[theorem]{Conjecture}

\hypersetup{
colorlinks = true,
linkcolor = black!30!blue,
citecolor = black!30!green
}

\title{Introducing \lop-kernels: a framework for kernelization lower bounds}

\author{J\'ulio Ara\'ujo}{Departamento de Matem\'atica, Universidade Federal do Cear\'a, Fortaleza, Brazil}{julio@mat.ufc.br}{https://orcid.org/0000-0001-7074-2753}{CNPq-Pq 304478/2018-0, CAPES-PrInt 88887.466468/2019-00 and CAPES-STIC-AmSud 88881.569474/2020-01.}

\author{Marin Bougeret}{LIRMM, Universit\'e de Montpellier, CNRS, Montpellier, France}{marin.bougeret@lirmm.fr}{https://orcid.org/0000-0002-9910-4656}{}

\author{Victor Campos}{Departamento de Computa\c c\~ao, Universidade Federal do Cear\'a, Fortaleza, Brazil}{victoitor@ufc.br}{https://orcid.org/0000-0002-2730-4640}{FUNCAP - PNE-011200061.01.00/16.}

\author{Ignasi Sau}{LIRMM, Universit\'e de Montpellier, CNRS, Montpellier, France}{ignasi.sau@lirmm.fr}{https://orcid.org/0000-0002-8981-9287}{DEMOGRAPH (ANR-16-CE40-0028), ESIGMA (ANR-17-CE23-0010), ELIT (ANR-20-CE48-0008-01), and French-German Collaboration ANR/DFG Project UTMA (ANR-20-CE92-0027).}

\authorrunning{J\'ulio Ara\'ujo, Marin Bougeret, Victor Campos, and Ignasi Sau} 

\Copyright{J\'ulio Ara\'ujo, Marin Bougeret, Victor A. Campos, and Ignasi Sau} 



\category{} 

\relatedversion{Some of the results of this article appeared in the preliminary version published in the proceedings of IPEC 2021~\cite{MMVC_IPEC}. This article is available at \texttt{https://arxiv.org/abs/2102.02484}. } 

\supplement{}

\nolinenumbers 

\hideLIPIcs  

\EventNoEds{2}
\EventLongTitle{}
\EventShortTitle{arXiv preprint}
\EventAcronym{}
\EventYear{2020}
\EventDate{August 24--28, 2020}
\EventLocation{Prague, Czech Republic}
\EventLogo{}
\SeriesVolume{}
\ArticleNo{1}

\begin{document}

\maketitle

\begin{abstract}
In the {\sc Maximum Minimal Vertex Cover} ({\sc MMVC}) problem, we are given a graph $G$ and a positive integer $k$, and the objective is to decide whether $G$ contains a minimal vertex cover of size at least $k$.
Motivated by the kernelization of {\sc MMVC} with parameter $k$, our main contribution is to introduce a simple general framework to obtain kernelization lower bounds for a certain type of kernels for optimization problems, which we call \emph{\lop-kernels}. Informally, this type of kernels is required to preserve large optimal solutions in the reduced instance, and captures the vast majority of existing kernels in the literature.

As a consequence of this framework, we show that the trivial quadratic kernel for {\sc MMVC} is essentially optimal, answering a question of Boria et al.~[Discret. Appl. Math.~2015], and that the known cubic kernel for {\sc Maximum Minimal Feedback Vertex Set} is also essentially optimal.  We present further applications for \textsc{Tree Deletion Set} and for \textsc{Maximum Independent Set} on $K_t$-free graphs.

Back to the \textsc{MMVC} problem, given the (plausible) non-existence of subquadratic kernels for {\sc MMVC} on general graphs, we provide subquadratic kernels on $H$-free graphs for several graphs $H$, such as the bull, the paw, or the complete graphs, by making use of the Erd\H{o}s-Hajnal property. Finally, we prove that {\sc MMVC} does not admit polynomial kernels parameterized by the size of a minimum vertex cover of the input graph, even on bipartite graphs, unless ${\sf NP} \subseteq {\sf coNP} / {\sf poly}$.
\end{abstract}

\newpage

\section{Introduction}
\label{sec:intro}



A \emph{vertex cover} in a graph $G$ is a subset of vertices containing at least one endpoint of every edge. In the associated optimization problem, called \textsc{Minimum Vertex Cover}, the objective is to find, given an input graph $G$, a vertex cover in $G$ of minimum size. This problem has been one of the leitmotifs of the area of parameterized complexity~\cite{CyganFKLMPPS15,DF13}, serving as a test bed for many of the most fundamental techniques. An instance of a \emph{parameterized problem} is of the form $(x,k)$, where $x$ is the total input (typically, a graph) and $k$ is a positive integer called the \emph{parameter}. The crucial notion is that of \emph{fixed-parameter tractable} algorithm, \FPT for short, which is an algorithm deciding whether $(x,k)$ is a positive instance in time $f(k) \cdot |x|^{\Ocal(1)}$, where $f$ is a computable function depending only on $k$. In the parameterized \textsc{Vertex Cover} problem, we are given a graph $G$ and an integer parameter $k$, and the objective is to decide whether $G$ contains a vertex cover of size at most $k$. One of the main fields within parameterized complexity is \emph{kernelization}~\cite{Book-kernels}, where the objective is to decide whether an instance $(x,k)$ of a parameterized problem can be transformed in polynomial time into an equivalent instance $(x',k')$ whose total size is bounded by a function of $k$; the reduced instance is called a \emph{kernel}, and finding kernels of small size, typically polynomial or even linear in $k$ in the best case, is one of the most active areas of parameterized complexity.

Considering the ``max-min'' version of minimization problems, that is, maximizing the size of a {\sl minimal} solution of the corresponding problem, is a natural approach that has been applied to several problems such as \textsc{Dominating Set}~\cite{BazganBCFJKLLMP18,DubloisLP21} (whose ``max-min'' version is called \textsc{Upper Domination}), \textsc{Feedback Vertex Set}~\cite{DubloisHGLM20}, or \textsc{Hitting Set}~\cite{Damaschke11,MMBS-paper}. The initial motivation of this article is the ``max-min'' version of \textsc{Minimum Vertex Cover}, called \textsc{Maximum Minimal Vertex Cover}, or just \textsc{MMVC} for short.

\medskip
\noindent \textbf{Previous work}. In his habilitation, Fernau~\cite{FernauHDR} presented \FPT algorithms for  \textsc{MMVC} as well as some results about its kernelization parameterized by the solution size $k$. It is easy to note, as observed in~\cite{FernauHDR},  that the problem admits a kernel with at most $k^2$ vertices: if some vertex has degree at least $k$, we can safely answer ``\yes'' (cf. \autoref{lem:extension-nbh-is} for a proof); otherwise, the maximum degree is at most $k-1$, and it follows that every instance without isolated vertices (which may be safely removed) that has at least $k^2$ vertices is a \yes-instance, hence we have a trivial kernel with at most $k^2$ vertices. Fernau~\cite{FernauHDR}  presented a kernel with at most $4k$ vertices for \textsc{MMVC} restricted to planar instances using the algorithmic version of the Four Color Theorem~\cite{RobertsonSST97}, and claimed in~\cite[Corollary 4.25]{FernauHDR} a kernel with at most $2k$ vertices on general graphs using spanning trees. Unfortunately, this latter kernelization algorithm is incorrect, as we discuss in \autoref{ap:Fernau}.

Boria et al.~\cite{BoriaCP15} initiated a study of the complexity of  \textsc{MMVC} and presented a number of results, in particular a polynomial-time approximation algorithm with ratio $n^{1/2}$ on $n$-vertex graphs, and showed that, unless $\P = \NP$, no polynomial-time approximation algorithm with ratio $n^{1/2 - \varepsilon}$ exists for any $\varepsilon > 0$. They also presented \FPT algorithms for \textsc{MMVC} for several choices of the parameters such as the treewidth, the size of a maximum matching, or the size of a minimum vertex cover of the input graph. The authors asked explicitly whether kernels of size $o(k^2)$ exist for \textsc{MMVC} parameterized by $k$.

Zehavi~\cite{Zehavi17} presented tight \FPT algorithms, under the Strong Exponential Time Hypothesis, for \textsc{MMVC} and its weighted version parameterized by the size of a minimum vertex cover. Recently, Bonnet and Paschos~\cite{BonnetP18} and Bonnet et al.~\cite{BonnetLP18} considered the inapproximability of \textsc{MMVC} in subexponential time.

Note that the \textsc{MMVC} problem is the dual of the well-studied \textsc{Minimum Independent Dominating Set} problem (to see this, note that the complement of any minimal vertex cover is an independent dominating set), which has applications in wireless and ad-hoc networks~\cite{HurinkN08}. We refer to the survey of Goddard and Henning~\cite{GoddardH13}.

\medskip
\noindent \textbf{Our results and techniques}. The starting motivation of this article is the kernelization of the \textsc{MMVC} problem, which has been almost unexplored so far in the literature. This initial motivation has resulted in a general framework that can be applied to a broad class of optimization problems in order to derive kernelization lower bounds.

Namely,
motivated by the question of Boria et al.~\cite{BoriaCP15} about the existence of subquadratic kernels for \textsc{MMVC}, we introduce a generic framework to obtain kernelization lower bounds for a ``certain type'' of kernels for parameterized maximization or minimization problems (in particular, for \textsc{MMVC}), based on a hypothesis that guarantees an inapproximability result, typically $\P \neq\NP$. Informally, by ``certain type'' we mean kernelization algorithms that, in polynomial time, either decide the instance (by answering ``\yes'' or ``\no'') or produce an equivalent instance of the considered problem in which the value of an optimal solution is ``preserved'', in the sense that it may drop only by the drop suffered by the parameter; see \autoref{sec:lop-max-def} for the formal details for the case of maximization problems. We call such kernels \emph{large optimal preserving kernels}, or \emph{\lop-kernels} for short. Even if this type of kernels may seem restrictive, in particular we are not aware of any ``non-artificial'' kernel for a maximization problem, such as those that have become nowadays standard~\cite{Book-kernels}, which is {\sl not} a \lop-kernel. (We do have such an example for a minimization problem, as discussed later.)
The core idea of our approach is to show that a \lop-kernel yields a polynomial-time approximation algorithm whose ratio depends on the size (and most importantly, on the degree) of the kernel, and to use known inapproximability results to obtain the desired lower bound.

We present our framework of \lop-kernels separately for maximization (\autoref{sec:framework-max}) and minimization (\autoref{sec:framework-min}) problems. Even if both versions are similar, they are {\sl not} totally symmetric, and a number of technical differences pop up;
we discuss them in detail as they appear in \autoref{sec:framework-min}. Our general result is stated in \autoref{thm:lop-kernelsG} and \autoref{thm:lop-kernelsGmin} for maximization and minimization problems, respectively. In order be able to apply our framework to an optimization problem, we need it to be ``well-behaved'', a mild condition defined in \autoref{sec:framework-max} and \autoref{sec:framework-min} that, for instance, for vertex-optimization problems is weaker than their decision version being in \NP. Also, our results distinguish the existence of constructive or non-constructive approximation algorithms. Since our framework seems to particularly fit vertex-optimization problems, we present the particular cases of \autoref{thm:lop-kernelsG} and \autoref{thm:lop-kernelsGmin} for vertex-maximization and vertex-minimization problems in \autoref{thm:lop-kernels} and \autoref{thm:lop-kernelsmin}, respectively. In order to ease the application of our results to concrete problems, we provide in \autoref{cor:lop-kernels} and \autoref{cor:lop-kernelsmin} the ``contrapositive'' versions of \autoref{thm:lop-kernels} and \autoref{thm:lop-kernelsmin}, respectively.

\medskip
\noindent
\emph{Applications of our framework}. Combining \autoref{cor:lop-kernels} with the $\O(n^{\frac{1}{2}-\varepsilon})$-inapproximability result for {\sc MMVC} by Boria et al.~\cite{BoriaCP15} immediately rules out (cf. \autoref{cor:lop1}) the existence of a \lop-kernel for {\sc MMVC} with $\Ocal(k^{2 - \varepsilon})$ vertices for any $\varepsilon > 0$, unless $\P = \NP$. Thus, while \autoref{cor:lop1} does not completely rule out the existence of subquadratic kernels for \textsc{MMVC}, it tells that, if such a kernel exists, it should consist of ``non-standard'' reduction rules.

Interestingly, our framework has consequences beyond the {\sc MMVC} problem. One of them concerns the {\sc Maximum Minimal Feedback Vertex Set} ({\sc MMFVS}) problem, defined in the natural way. Dublois et al.~\cite{DubloisHGLM20} recently provided a cubic kernel for {\sc MMFVS} parameterized by the solution size, and proved that the problem does not admit an $\O(n^{\frac{2}{3}-\varepsilon})$-approximation algorithm for any $\varepsilon>0$, unless $\P = \NP$. By applying \autoref{cor:lop-kernels}  we directly  obtain (\autoref{cor:lop2}) that the cubic kernel of Dublois et al.~\cite{DubloisHGLM20} is ``essentially'' optimal.

Another application of our results deals with the \textsc{Tree Deletion Set} problem. In this case, the fact that this problem does not admit a polynomial-time $\O(n^{1-\varepsilon})$-approximation for any $\varepsilon > 0$ unless $\P \neq\NP$~\cite{Yannakakis79} implies, together with \autoref{cor:lop-kernelsmin}, that {\sc Tree Deletion Set} parameterized by the solution size does not admit a polynomial \lop-kernel, unless $\P = \NP$ (\autoref{cor:lop3}). However, \textsc{Tree Deletion Set} {\sl does} admit a polynomial kernel with $\O(k^4)$ vertices~\cite{GiannopoulouLSS16}.
Therefore, this polynomial kernel cannot be a \lop-kernel, and so far it constitutes the only non-artificial example of non-\lop-kernel that we are aware of.

Our last application concerns the \textsc{Maximum Independent Set} problem restricted to $K_t$-free graphs. In particular, we show (\autoref{cor:lop4}) that  a \lop-kernel with $\Ocal(k^{t -1- \varepsilon})$ vertices for {\sc Maximum Independent Set} on $K_t$-free graphs would improve the best known approximation ratio $n^{\frac{t-2}{t-1}}$ that follows from Ramsey's theorem~\cite{Ramsey}. Finally, generalizing a conjecture of Bonnet et al.~\cite{BonnetTTW20}, we conjecture that for every fixed graph $H$, the {\sc Maximum Independent Set} problem restricted to $H$-free graphs admits a polynomial \lop-kernel.


\medskip
\noindent \emph{Comparison with other frameworks.} Compared to existing frameworks to obtain lower bounds on kernelization, such as cross-compositions~\cite{BodlaenderDFH09,BodlaenderJK14}, weak compositions~\cite{DellM12,DellM14,HermelinW12}, polynomial parameter transformations~\cite{BodlaenderTY11, Binkele-RaibleFFLSV12}, or techniques to obtain lower bounds on the coefficients of linear kernels~\cite{ChenFKX07},
or that relate approximation and kernelization~\cite{GuoKK11,LokshtanovPRS17,BiswasRS20,Kratsch12,Abu-KhzamBCF14}, our approach has the advantages that it is quite simple, straightforward to apply, and relies on the same hypothesis on which the corresponding inapproximability result is based, typically the standard hypothesis that $\P \neq \NP$.  On the negative side, it has the following two drawbacks. The first one is that, in order to obtain a non-trivial lower bound on the kernel size,  it can only be applied to  problems which are quite hard to approximate, for example within a factor $\Ocal(n^{r - \varepsilon})$ for some constant $r >0$, as it is the case of \textsc{MMVC} and \textsc{MMFVS}. The second, and probably most important, drawback of our techniques is that they are able to rule out the existence of what we call \lop-kernels of certain sizes, but smaller non-standard kernels that do not preserve the value of large optimal solutions might, a priori, still exist (as it is the case for \textsc{Tree Deletion Set}, as discussed above). Hence, since our framework  seems to be orthogonal to existing ones, we think that it adds to the above list of techniques to obtain kernelization lower bounds.

\medskip
\noindent \emph{Other results on the kernelization of {\sc MMVC}.}
Coming back to the \textsc{MMVC} problem parameterized by the solution size, given the above negative result on general graphs, we identify graph classes where  \textsc{MMVC} is still \NP-hard and admits a subquadratic kernel. In particular, we deal with graph classes defined by excluding an {\sl induced} subgraph $H$ that satisfies the \emph{Erd\H{o}s-Hajnal property}~\cite{ErdosH89}, that is, for which there exists a constant $\delta>0$ such that every $H$-free graph on $n$ vertices contains either a clique or an independent set of size $n^{\delta}$. In particular, we present a kernel for \textsc{MMVC} with $\Ocal(k^{7/4})$ vertices on the well-studied class of bull-free graphs (\autoref{thm:kernel-bull-free}), with $\Ocal(k^{\frac{2t-3}{t-1}})$ vertices on $K_t$-free graphs graphs for every $t \geq 3$ (\autoref{thm:kernel-Kt-free}), and with $\Ocal(k^{5/3})$ vertices on paw-free graphs (\autoref{thm:kernel_paw}). To the best of our knowledge, this is the first time that the Erd\H{o}s-Hajnal property is used to obtain polynomial kernels (we would like to note that it was used by Kratsch et al.~\cite{KratschPRR14} to obtain kernelization lower bounds).

Our strategy to obtain these subquadratic kernels on $H$-free graphs is as follows. By the high-degree rule mentioned above, given an instance $(G,k)$, we may assume that the maximum degree of $G$ is at most $k-1$. We find greedily a minimal vertex cover $X$ of $G$. If $|X| \geq k$ we are done, so we may assume that $|X| \leq k-1$, hence the goal is to bound the size of $S:= V(G) \setminus X$. Using that $G[X]$ is also $H$-free, the Erd\H{o}s-Hajnal property implies (\autoref{lem:EH-partition}) that $X$ can be partitioned in polynomial time into a sublinear (in $k$) number of independent sets and cliques. Since $S$ is an independent set and we may assume that $G$ has no isolated vertices, in order to bound $|S|$ by a subquadratic function of $k$, it is enough to show that, for each of the sublinearly many cliques or independent sets $Y$ that partition $X$, its neighborhood in $S$ has size $\Ocal(k)$. This is easy if $Y$ is an independent set: if $|N_S(Y)| \geq k$ we can conclude that $(G,k)$ is a \yes-instance (\autoref{lem:extension-nbh-is}), so we may assume that $|N_S(Y)| \leq k-1$. The case where $Y$ is a clique is more interesting, and we need ad-hoc arguments depending on each particular excluded induced subgraph $H$.

We also present several positive results for \textsc{MMVC} restricted to other particular graph classes, such as $K_{1,t}$-free graphs (\autoref{lem:kernel_K1t}), graph classes with bounded chromatic number (\autoref{lem:polyChi}), or graph classes with bounded cliquewidth (\autoref{obs:kernel_MSO}).

Finally, we show (\autoref{thm:no-poly-kernel})
that \textsc{MMVC}, parameterized by the size of a minimum vertex cover (or of a maximum matching) of the input graph, does not admit a polynomial kernel unless ${\sf NP} \subseteq {\sf coNP} / {\sf poly}$, even restricted to bipartite graphs.
This result complements the \FPT algorithms for \textsc{MMVC} under these parameterizations given by Boria et al.~\cite{BoriaCP15} and Zehavi~\cite{Zehavi17}, and shows that, in what concerns the existence of polynomial kernels for \textsc{MMVC}, the most natural structural parameters smaller than the solution size are not large enough to yield polynomial kernels (note that the treewidth of any graph is at most one more than its vertex cover number, hence our result rules out the existence of polynomial kernels for \textsc{MMVC} parameterized by treewidth as well). The proof consists of a polynomial parameter transformation from \textsc{Monotone Sat}  parameterized by the number of variables. In particular, our reduction yields also the \NP-hardness of  \textsc{MMVC} on bipartite graphs, which provides an alternative proof to the one of Boliac and Lozin~\cite{BoliacL03} via the \NP-hardness  of \textsc{Minimum Independent Dominating Set} on bipartite graphs.

\medskip
\noindent \textbf{Organization}. In \autoref{sec:prelim} we provide some basic preliminaries about graphs, the {\sc MMVC} problem, parameterized complexity, and approximation algorithms.
In \autoref{sec:framework-max} (resp. \autoref{sec:framework-min}) we present our framework to obtain kernelization lower bounds for maximization (minimization) problems. In both sections, the contents are split into three subsections: we start with the general definitions, then we focus on the particular and relevant case of vertex-optimization problems, and then we present the general results for what we call ``well-behaved'' optimization problems. In \autoref{sec:applications-lop} we present several applications of the framework of \lop-kernels for concrete problems, and in \autoref{ap:Fernau} we discuss the flaw in the linear kernel for \textsc{MMVC} claimed by Fernau~\cite{FernauHDR}. \autoref{sec:subquadratic-kernels-MMVC} is devoted to the subquadratic kernels for \textsc{MMVC} on particular graph classes, as well as to other positive results for \textsc{MMVC}. Our reduction to rule out
 the existence of polynomial kernels for \textsc{MMVC} parameterized by the size of a minimum vertex cover (or a maximum matching) is presented in \autoref{sec:nopolykernel}.  We conclude the article in \autoref{sec:concl} with a discussion and some directions for further research.



%
%

\section{Preliminaries}
\label{sec:prelim}



\noindent\textbf{Graphs and functions.} We use standard graph-theoretic notation, and we refer the reader to~\cite{Diestel12} for any undefined notation. For an integer $p \geq 1$, we let $[p]$ be the set containing all integers $i$ with $1 \leq i \leq p$. We use $\uplus$ to denote  the disjoint union. We will only consider finite undirected graphs without loops nor multiple edges, and we denote an edge between two vertices $u$ and $v$ by $\{u,v\}$. A subgraph $H$ of a graph $G$ is \emph{induced} if $H$ can be obtained from $G$ by deleting a set of vertices $D = V(G) \setminus S$, and we denote $H = G[S]$. Given a graph $H$, a graph $G$ is \emph{$H$-free} if it does not contain any induced subgraph isomorphic to $H$.  If ${\cal H}$ is a collection of graphs, a graph $G$ is \emph{${\cal H}$-free} if it is $H$-free for every $H \in {\cal H}$. For a graph $G$ and a set $S \subseteq V(G)$, we use the notation $G \setminus S =G[V(G) \setminus S]$, and for a vertex $v \in V(G)$, we abbreviate $G \setminus \{v\}$ as $G \setminus v$. A vertex $v$ is \emph{complete} to a set $S \subseteq V(G)$ if $v$ is adjacent  to every vertex in $S$.

The \emph{open} (resp. \emph{closed}) \emph{neighborhood} of a vertex $v$ in a graph $G$ is denoted by $N(v)$ (resp. $N[v]$), whenever the graph $G$ is clear from the context. For vertex sets $X,Y \subseteq V(G)$, we define $N[X] = \bigcup_{v \in X}N[v]$, $N(X) = N[X] \setminus X$, $N_Y[X] = N[X] \cap Y$, and $N_Y(X) = N_Y[X] \setminus X$. The \emph{degree} of a vertex $v$ in a graph $G$ is defined as $|N(v)|$, and we denote it by $\deg_G(v)$, or just $\deg(v)$ of the graph is clear from the context.
For an integer $t \geq 1$, we denote by $P_t$ (resp. $I_t$, $K_t$) the path (resp. edgeless graph, complete graph) on $t$ vertices.
For two integers $a,b\geq 1$, we denote by $K_{a,b}$ the bipartite graph with parts of sizes $a$ and
$b$. 

A \emph{clique} (resp. \emph{independent set}) in a graph $G$ is a set of vertices that are pairwise adjacent (resp. not adjacent). A graph property is \emph{hereditary} if whenever it holds for a graph $G$, it holds for all its induced subgraphs as well. Note that the properties of being an edgeless graph, a complete graph, or an independent set are hereditary.  We denote by $\Delta(G)$ (resp. $\omega(G)$ the maximum vertex degree (resp. clique size) of a graph $G$.

A \emph{vertex cover} of a graph $G$ is a set of vertices containing at least one endpoint of every edge, and it is \emph{minimal} if no proper subset of it is a vertex cover. One of the concrete problems that we study in this paper is formally stated as follows. We state it as a decision problem, since most of our results consider its parameterization by the solution size $k$.

\vspace{-.5cm}
\probl{Maximum Minimal Vertex Cover (MMVC)}{A graph $G$ and a positive integer $k$.}{Does $G$ contain a minimal vertex cover of size at least $k$?}

For a graph $G$, we denote by  ${\sf mmvc}(G)$ the maximum size of a minimal vertex cover of $G$. The following observation has been already used in previous work~\cite{BoriaCP15,Zehavi17}.
\begin{observation}\label{obs:characterization-mvc}
Let $G$ be a graph. A set $X \subseteq V(G)$ is a minimal vertex cover of $G$ if and only if $X$ is a vertex cover of $G$ and, for every vertex $v \in X$, $N(v) \nsubseteq X$.
\end{observation}

The next lemma provides a useful way to conclude that we are dealing with a \yes-instance in the kernelization algorithms presented in \autoref{sec:subquadratic-kernels-MMVC}.

\begin{lemma}\label{lem:extension-nbh-is}
Let $G$ be a graph and let $S \subseteq V(G)$ be an independent set. There exists a minimal vertex cover of $G$ containing $N(S)$.
\end{lemma}
\begin{proof}
Note that, since $S$ is an independent set,  $V(G) \setminus S$ is a vertex cover of $G$. Hence, there exists a minimal vertex cover $X$ of $G$ such that $X \subseteq V(G) \setminus S$. We claim that $N(S) \subseteq X$. Suppose for the sake of contradiction that there exists a vertex $v \in N(S)$ such that $v \notin X$. Since $v$ has a neighbor $u$ in $S$ and $S \cap X = \emptyset$, the edge $\{u,v\}$ would not be covered by $X$.
\end{proof}

Note that, in particular, \autoref{lem:extension-nbh-is} implies that if $(G,k)$ is an instance of the {\sc Maximum Minimal Vertex Cover} problem and $v \in V(G)$ is a vertex of degree at least $k$, then we can conclude that  $(G,k)$ is a \yes-instance. This will allow us to assume, in our kernelization algorithms,  that $\Delta(G) \leq k-1$.

\medskip
\noindent
\textbf{Parameterized complexity.} We refer the reader to~\cite{DF13,CyganFKLMPPS15} for basic background on parameterized complexity, and we recall here only some basic definitions used in this article. A \emph{parameterized problem} is a language $L \subseteq \Sigma^* \times \mathbb{N}$.  For an instance $I=(x,k) \in \Sigma^* \times \mathbb{N}$, $k$ is called the \emph{parameter}.

A parameterized problem is \emph{fixed-parameter tractable} (\textsf{FPT}) if there exists an algorithm $\mathcal{A}$, a computable function $f$, and a constant $c$ such that given an instance $I=(x,k)$, $\mathcal{A}$   (called an \textsf{FPT} \emph{algorithm}) correctly decides whether $I \in L$ in time bounded by $f(k) \cdot |I|^c$. For instance, the \textsc{Vertex Cover} problem parameterized by the size of the solution is \textsf{FPT}.

%

For an instance $(x, k)$ of a parameterized problem $Q$, a \emph{kernelization algorithm} is an algorithm $\mathcal{A}$ that, in polynomial time, generates from $(x, k)$ an equivalent instance $(x', k')$ of $Q$ such that $|x'| + k' \leq f(k)$, for some computable function $f : \mathbb{N} \rightarrow \mathbb{N}$, where $|x'|$ denotes the size of $x'$. If $f(k)$ is bounded from above by a polynomial of the parameter, we say that $Q$ admits a \emph{polynomial kernel}. In particular, if $f(k)$ is bounded by a linear (resp. quadratic) function, then we say that $Q$ admits a \emph{linear} (resp. \emph{quadratic}) kernel.

A \emph{polynomial  parameter transformation}, abbreviated as \PPT, is an algorithm  that, given an instance $(x, k)$ of a parameterized problem $A$, runs in time polynomial  in $|x|$ and outputs an instance $(x', k')$ of a parameterized problem $B$ such that $k'$ is bounded from above by a polynomial on $k$ and $(x, k)$ is positive if and only if $(x',k')$ is positive. If a parameterized problem $A$ does not admit a polynomial kernel unless ${\sf NP} \subseteq {\sf coNP} / {\sf poly}$ and there exists a \PPT from $A$ to a parameterized problem $B$, then $B$ does not admit a polynomial kernel unless ${\sf NP} \subseteq {\sf coNP} / {\sf poly}$ either~\cite{CyganFKLMPPS15}.

\medskip
\noindent
\textbf{Approximation algorithms.}  We refer the reader to~\cite{WilliamsonS-book} for background on approximation algorithms, and we define here only some non-standard notions used in this article.

As, when dealing with graph problems, one typically measures the size of kernels in terms of the number of vertices or edges (and not in the classical bit-size of the instance), we introduce an arbitrary notion of size as follows. Given a optimization problem $\Pi$, we say that a non-negative integer-valued function $|\cdot|$ is a \emph{size function} if, given an instance $I$ of $\Pi$, $|I|$ can be computed in polynomial time in the classical bit-size, and $|I|$ is upper-bounded by a polynomial in the classical bit-size.

For an optimization problem $\Pi$, an instance $I$ of $\Pi$, and a feasible solution $s$ of $\Pi$ in $I$, we denote by $\val(I,s)$ the value of the objective function of $\Pi$ for $s$.  We restrict ourselves to optimization problems $\Pi$ whose objective functions for feasible solutions take non-negative integer values. For a maximization (resp. minimization) problem $\Pi$
and an instance $I$ of $\Pi$, we denote by $\opt(I)$ the maximum (resp. minimum) of $\val(I,s)$ over all feasible solutions $s$ of $\Pi$ in $I$.

A maximization (resp. minimization) problem $\Pi$ is a \emph{vertex-maximization} (resp. \emph{vertex-minimization}) problem if their instances consist of a graph $G$, and the objective is to find a vertex set $S \subseteq V(G)$ of maximum (resp. minimum) size satisfying some conditions. For instance, \textsc{Maximum Independent Set} and \textsc{Minimum Vertex Cover} are typical examples of vertex-maximization and vertex-minimization problems, respectively.

For kernelization purposes, we need to consider the \emph{decision} version of optimization problems. For a maximization (resp. minimization) problem $\Pi$ whose instances are of the form $I$, we denote by $\Pi_{\dec}$ the decision problem whose instances are of the form $(I,k)$, where $k$ is a non-negative integer, and where $(I,k)$ is a \emph{\yes-instance} of $\Pi_{\dec}$ if  $\opt(I) \geq k$ (resp. $\opt(I) \leq k$), and a \emph{\no-instance} otherwise.

Since we aim at establishing a link between the existence of certain kernels and approximation algorithms, we need to take care of {\sl constructibility} issues. The standard definition of a kernelization algorithm~\cite{Book-kernels} does not involve constructing a solution of the considered problem. On the other hand, the standard definition of an approximation algorithm~\cite{WilliamsonS-book} {\sl does} take into account the construction of the corresponding solution. Hence, in order to establish such a connection, we need to consider slightly ``non-standard'' definitions of these objects.

Namely, we distinguish between constructive and non-constructive approximation algorithms, so that the kernelization lower bounds that we present are able to rule out constructive or non-constructive kernels (cf. the second paragraph of \autoref{sec:framework-max}). To this end, we say that an algorithm for a maximization (resp. minimization) decision problem $\Pi_{\dec}$  \emph{constructively decides} an instance $(I,k)$ if, whenever it holds that $\opt(I) \geq k$ (resp. $\opt(I) \leq k$), the algorithm outputs a feasible solution $s$ such that $\val(I,s) \geq k$ (resp. $\val(I,s) \leq k$).

When using the term ``approximation algorithm'' with ratio $\rho \geq 1$ for a maximization (resp. minimization) problem $\Pi$, we assume, unless stated otherwise, that it is \emph{constructive}, that is, that the algorithm, given an instance $I$ of $\Pi$, outputs a feasible solution $s$ of $\Pi$ in $I$ such that $\opt(I)/\val(I,s) \leq \rho$ (resp. $\val(I,s) / \opt(I) \leq \rho$). Note that the approximation ratio $\rho$ is, in general, a non-negative integer-valued function that depends on $I$.

We define a \emph{value-approximation algorithm} with ratio $\rho \geq 1$ for a maximization (resp. minimization) problem $\Pi$ as an algorithm that, given an instance $I$ of $\Pi$, returns a non-negative integer $k$ such that $1 \leq \opt(I)/k \leq \rho$ (resp. $1 \leq k / \opt(I) \leq \rho$). Again, here $\rho$ is, in general, a non-negative integer-valued function that depends on $I$. Note that a value-approximation algorithm is not only not required to construct a feasible solution with value $k$, but also not required to guarantee that such a solution exists.


\section{A framework for ruling out certain polynomial kernels: the case of maximization problems}
\label{sec:framework-max}

In this section we introduce our generic framework to obtain lower bounds on the size of a certain type of polynomial kernels, which we call \emph{\lop-kernels} (see \autoref{def:lop-kernel}), for a broad class of maximization problems that we proceed to introduce. Informally, the framework is based on simple and self-contained arguments proving that a ``small'' \lop-kernel implies the existence of a ``good'' approximation algorithm. Then, the contrapositive of this statement implies that inapproximability results can be turned into \lop-kernel lower bounds.

It worth mentioning here that most of the inapproximability results in the literature hold for the {\sl value}-approximation algorithms as defined  at the end of \autoref{sec:prelim}, that is, for algorithms that are {\sl not} required  to construct in polynomial time an appropriate solution of the corresponding problem, but only to report a value within the appropriate range. In order to guarantee that it is also possible to use our framework when only the non-existence of {\sl constructive} approximation algorithms is known, we introduce a variant of \lop-kernels, called \emph{constructive \lop-kernels}, such that their existence implies the existence of a constructive approximation algorithm. However, in a first reading, we recommend to skip all technical details concerning constructibility.


We say that a maximization problem $\Pi$  is \emph{well-behaved} if it comes equipped with a size function (as defined in \autoref{sec:prelim}) and it satisfies the following condition, which we denote by $C^{\max}$:
\begin{quotation}
  \noindent There exists an algorithm that, given as input a real number $c$ and an instance  $I$ of $\Pi$ such that $\opt(I) \le c$,  runs in polynomial time for every fixed $c$ and either decides that $\opt(I)=0$ and provides a feasible solution $s$ with $\val(I,s)=0$, or provides a feasible solution~$s$ with $\val(I,s) > 0$.
\end{quotation}

Observe that most of the classical maximization problems are well-behaved, and in particular any vertex-maximization problem whose decision version belongs to $\NP$ is well-behaved, as we can enumerate all subsets of vertices of size at most $c$, and for each of them verify in polynomial time if it is a feasible solution.
Given a well-behaved maximization problem $\Pi$,
we say that a function $\ub: \mathds{N} \to \mathds{N}$ is an \emph{upper bound function} for $\Pi$ if for any instance $I$ of $\Pi$, it holds that $\opt(I) \le \ub(|I|)$, where $|\cdot|$ is the size function of $\Pi$.
Throughout the paper, we assume that the notions of size used in both the size of kernels and the upper bound function are the same.

The reminder of this section is organized as follows. In \autoref{sec:lop-max-def} we present the definition of \lop-rules and \lop-kernels for well-behaved maximization problems, and we prove a general technical result, namely \autoref{lem:lop-dual}.
In \autoref{sec:res-max} we present the connection between \lop-kernels and approximation algorithms for vertex-maximization problems, and in \autoref{sec:res-max-gnl}
 we generalize it to arbitrary well-behaved maximization problems.

\subsection{Definition of \lop-rules and \lop-kernels}
\label{sec:lop-max-def}

\begin{definition}\label{def:lop}
  A \emph{large optimal preserving} reduction rule, or \emph{\lop-rule} for short, for a well-behaved maximization problem $\Pi$, is a polynomial-time algorithm $R$ that, given an instance
  $(I,k)$ of $\Pi_{\dec}$, computes another instance $(I',k')$ of $\Pi_{\dec}$ with $0 \leq k' \leq k$ and such that
  \begin{enumerate}
    \item if $(I,k)$ is a \no-instance of $\Pi_{\dec}$, then $(I',k')$ is a \no-instance of $\Pi_{\dec}$, and
    \item\label{point3}  if $(I,k)$ is a  \yes-instance of $\Pi_{\dec}$, then $\opt(G') \geq \opt(G) - (k-k')$, implying that $(I',k')$ is a \yes-instance of $\Pi_{\dec}$.
  \end{enumerate}
  A \lop-rule $R$ is \emph{constructive} if, given $I$ and any solution $s'$ of $I'$ of such that $\val(I',s') \ge k'$, it constructs (in polynomial time) a solution $s$ of $I$ such that $\val(I,s) \ge k$.
\end{definition}

Note that Property~\ref{point3} in \autoref{def:lop} is stronger than the implication ``if $(I,k)$ is a \yes-instance  of $\Pi_{\dec}$, then $(I',k')$ is a \yes-instance  of $\Pi_{\dec}$'', which would yield the definition of a classical kernelization algorithm~\cite{CyganFKLMPPS15,DF13}.
Indeed, when we consider how this latter implication is generally proved in safeness proofs of classical kernels, one of the following scenarios often occurs:
\begin{enumerate}[(a)]
\item\label{case1} For every solution $s$ of $I$ there exists a solution $s'$ of $I'$ with $\val(I',s') \ge \val(I,s)-(k-k')$.
\item\label{case2} For every solution $s$ of $I$ with $\val(I,s) \ge k$, there exists a solution $s'$ of $I'$ with $\val(I',s') \ge \val(I,s)-(k-k')$.
\item\label{case3} If there exists a solution $s$ of $I$ with $\val(I,s) \ge k$, then there exists a solution $s'$ of $I'$ with $\val(I',s') \ge  k'$.
\end{enumerate}

In Case~(\ref{case1}), the rule preserves all optimal values, as it implies that $\opt(G') \ge \opt(G)-(k-k')$.
In Case~(\ref{case2}), the rule preserves only large optimal values, as it implies that if $\opt(G) \ge k$, then $\opt(G') \ge \opt(G)-(k-k')$, implying Property~\ref{point3} above.
Note that if $\opt(G) < k$, then $\opt(G')$ and $\opt(G)$ are not necessarily related.  This justifies our choice for ``large optimal preserving'' rules.
Case~(\ref{case3}) corresponds to the weaker and classical implication ``if $(I,k)$ is a \yes-instance of $\Pi_{\dec}$, then $(I',k')$ is a \yes-instance of $\Pi_{\dec}$''.

The following observation is an immediate consequence of the definition of a \lop-rule.

\begin{observation}\label{prop:lop-chain}
  \lop-rules can be composed. Formally, consider two \lop-rules $R_1$ and $R_2$.
  Then, the rule $R$ that, given a instance $(I,k)$ of $\Pi_{\dec}$, returns $R_2(R_1(I,k))$, is also a \lop-rule. Moreover, if $R_2$ and $R_1$ are constructive, then $R$ is also constructive.
\end{observation}


A typical example of a \lop-rule for a vertex-maximization problem is when we can identify a ``dominant'' set of vertices that can be safely included into a solution. More precisely,
 consider a rule that, given a graph $G$, finds a subset $T \subseteq V(G)$ and a graph $G'$ such that there exists an optimal solution $S^{\star}$ in $G$ such that $S^{\star} = T \cup S'$, where $S'$ is a solution in $G'$, and  for every solution $S'$ in $G'$, $S' \cup T$ is a solution in $G$.
Such a rule is a (constructive) \lop-rule, as we even fall into Case~(\ref{case1}) described above.

Even if we are not aware of known reduction rules for vertex-maximization problems that are {\sl not} \lop-rules, we can artificially devise such an example. For instance, for the \textsc{MMVC} problem, given an instance $(G,k)$, if there is a vertex that has more than $k$ neighbors of degree one, we can safely delete all but  any $k$ of them to obtain a reduced graph $G'$, and leave $k$ unchanged. Note that this rule falls into Case~(\ref{case3}) above, since by \autoref{lem:extension-nbh-is} both $G$ and $G'$ are \yes-instances of \textsc{MMVC}, but it does not satisfy Property~\ref{point3} in \autoref{def:lop}, since ${\sf mmvc}(G)$ may be arbitrarily larger than ${\sf mmvc}(G')$.

If we defined a \lop-kernel as an  algorithm consisting only of \lop-rules, we would exclude from being a \lop-kernel, for instance, a rule that detects a \yes-instance as in the above paragraph. This justifies the next definition, where we also allow \lop-kernels to decide instances.



\begin{definition}\label{def:lop-kernel}
 Let $\Pi$ be a well-behaved maximization problem and let $s: \mathds{N} \to  \mathds{N}$ be a computable function. A \emph{\lop-kernel of size $s$} for $\Pi$ parameterized by the solution size  is
 a polynomial-time algorithm that takes as input an instance $(I,k)$ of $\Pi_{\dec}$, and either
  \begin{itemize}
    \item decides that $(I,k)$ is a \yes-instance or a \no-instance, or
    \item outputs a reduced instance $(I',k')$ by applying a sequence of \lop-rules to $(I,k)$, with $|I'| \leq s(k)$.
  \end{itemize}
A \lop-kernel is \emph{constructive} if, in the first case, it constructively decides $(I,k)$ (but in the second case it may not use constructive rules).
\end{definition}

As it is common in kernels to exhaustively apply reduction rules, and \emph{then} to either decide the reduced instance or to output it, let us introduce and discuss the following definition\footnote{This definition was used in the conference version of this paper~\cite{MMVC_IPEC}.}.

\begin{definition}\label{def:lop-kernel-v2}
 Let $\Pi$ be a well-behaved maximization problem and let $s: \mathds{N} \to  \mathds{N}$ be a computable function. A \emph{\lop-kernel$^\star$ of size $s$} for $\Pi$ parameterized by the solution size  is
 a polynomial-time algorithm that takes as input an instance $(I,k)$ of $\Pi_{\dec}$, computes an instance $(I',k')$ by applying a (possibly empty) sequence of \lop-rules to $(I,k)$, and either
  \begin{itemize}
    \item decides that $(I',k')$ is a \yes-instance or a \no-instance, or
    \item outputs $(I',k')$, with $|I'| \leq s(k)$.
  \end{itemize}
A \lop-kernel$^\star$ is \emph{constructive} if, in the first case, it constructively decides $(I',k')$ and, in the second case, it only uses constructive \lop-rules.
\end{definition}

Firstly, observe that a \lop-kernel$^\star$ (resp. constructive \lop-kernel$^\star$) is a \lop-kernel (resp. constructive \lop-kernel).
Indeed, if a \lop-kernel$^\star$ decides $(I',k')$, then, as the definition of \lop-rules implies that the reduced instance $(I',k')$ is equivalent to $(I,k)$, it also decides $(I,k)$.
Moreover, if a \lop-kernel$^\star$ is constructive and decides that $(I',k')$ is a \yes-instance by providing a solution $s'$ with $\val(I',s') \ge k'$,
then, as the rules are constructive, and according to \autoref{prop:lop-chain}, we can build in polynomial time a solution $s$ with $\val(I,s) \ge k$, and thus
constructively decide $(I,k)$.
Secondly, observe that a \lop-kernel is a \lop-kernel$^\star$, but that a constructive \lop-kernel is not necessarily a constructive \lop-kernel$^\star$.
The conclusion of this discussion is that for the non-constructive versions, both definitions are equivalent, and for the constructive versions, \lop-kernels are slightly more general.
As many inapproximability results even hold for the non-constructive version of approximation, we suggest the reader to stick to the non-constructive version, and thus to chose any of the two definitions.
The only case where it could make a difference would be for a problem $\Pi$ for which inapproximability results are only known for ruling out constructive approximation algorithms. Then, \autoref{cor:lop-kernels} will turn this inapproximability
into a kernel lower bound even for constructive \lop-kernels, and not only for constructive \lop-kernels$^\star$.
This justifies why we consider  henceforth only \lop-kernels.

Our next objective is to prove that a \lop-kernel yields an approximation algorithm.
For this, we need the following definition, which is inspired by a similar notion introduced by Hochbaum and Shmoys~\cite{HochbaumS87}, and referred to as \emph{$f$-relaxed decision procedure} in~\cite{WilliamsonS-book}. 

 \begin{definition}\label{def:dual-approx}
   Let $\Pi$ be a well-behaved maximization problem and let $f: \mathds{N} \to \mathds{N}$  be a function.
     An \emph{$f$-dual-approximation algorithm for $\Pi$}  is a polynomial-time algorithm that, given an instance $(I,k)$ of $\Pi_{\dec}$, concludes one of the following:
   \begin{itemize}
   \item $\opt(I) \geq k$.
   \item $\opt(I) < f(k)$.
   \end{itemize}
 An $f$-dual approximation algorithm is \emph{constructive} if, whenever it concludes that $\opt(I) \geq k$, it provides a solution  $s$ with $\val(I, s) \ge k$.
 \end{definition}

 In the next lemma we prove that a \lop-kernel of size $s$ yields an $f$-dual-approximation algorithm (where $f$ depends on $s$), which in turn yields a classical approximation algorithm whose ratio depends on $s$. To provide some insight on the statement of the next lemma, keep in mind that for vertex-maximization problems, the upper bound function $u$ is typically the identity function.

 \begin{lemma}\label{lem:lop-dual}
	Let $\Pi$ be a well-behaved maximization problem with a non-decreasing upper bound function $\ub$ and let $s: \mathds{N} \to \mathds{N}$ be a computable function.
	If $\Pi$ admits a \lop-kernel of size $s$, then $\Pi$ admits an $f$-dual-approximation algorithm where $f(k) := \ub(s(k)) + k + 1$.
	Moreover, if the \lop-kernel is constructive, then the $f$-dual-approximation algorithm is also constructive.
\end{lemma}
\begin{proof}
	Let $k\in \mathds{N}$, $(I,k)$ be an instance of $\Pi_{\dec}$, and $\Rcal$ be a \lop-kernel of size $s$ for $\Pi$.
	We describe an $f$-dual approximation algorithm $\Acal$ which takes as input $I$ and $k$, starts by running $\Rcal$ with input $(I, k)$, and continues based on its possible output.
	If $\Rcal$ decides that $(I,k)$ is a \yes-instance, then $\opt(I)\ge k$, and $\Acal$ returns $\opt(I)\ge k$ as well.
        Notice that if $\Rcal$ is constructive, then it provides a solution $s$ of $I$ such that $\val(I,s) \ge k$, and $\Acal$ returns this solution as well.
	Otherwise, we claim that it is safe for $\Acal$ to return $\opt(I) \le  f(k)$.
	Indeed, if $\Rcal$ decides that $(I,k)$ is a \no-instance, then $\opt(I)< k$, implying $\opt(I) < \ub(s(k)) + k = f(k)$.
	Finally, suppose that $\Rcal$ outputs an equivalent instance $(I',k')$ obtained from $(I, k)$ using only \lop-rules and such that $|I'| \leq s(k)$.
	By using \autoref{prop:lop-chain} we can assume that $(I',k')$ is obtained from $(I,k)$ by a single \lop-rule, and Property~\ref{point3} in \autoref{def:lop} implies that $\opt(I) \le \opt(I')+(k-k') \le \ub(|I'|) + k  \leq   \ub(s(k)) + k < f(k)$, where we have used the fact that $\ub$ is non-decreasing.
\end{proof}



Let us now turn to our main results relating the size of \lop-kernels to the existence of approximation algorithms.
To keep statements as simple as possible, we first provide in \autoref{sec:res-max} results that correspond to the specialized versions for vertex-maximization problems of the general results presented in \autoref{sec:res-max-gnl}.

\subsection{Connection between \lop-kernels and approximation algorithms for vertex-maximization problems}
\label{sec:res-max}

In this subsection we deal with vertex-maximization problems. The following lemma is a folklore result~\cite{WilliamsonS-book}, but as it is generally tuned for a particular function $f$ appearing in the considered context, we need to restate it in a general form.

\begin{lemma}\label{lem:dual-approx}
  Let $\Pi$ be a vertex-maximization problem whose decision version is in \NP, and $f: \mathds{N} \to \mathds{N}$ be a computable function.
  \begin{enumerate}
  \item For every real number $c > 1$, if $\Pi$ admits an $f$-dual-approximation algorithm with $f(k) = \Ocal(k^c)$, then $\Pi$ admits a polynomial-time value-approximation algorithm with ratio $\O(n^{\frac{c-1}{c}})$ on $n$-vertex graphs.
  \item For every real number $\beta \ge 1$, if $\Pi$ admits an $f$-dual-approximation algorithm with $f(k) = \beta k+1$, then $\Pi$ admits a polynomial-time value-approximation algorithm with ratio $\beta+\varepsilon$
    for every real number $\varepsilon > 0$.
  \end{enumerate}
  Moreover, if the $f$-dual-approximation is constructive, then the corresponding approximation algorithm is also constructive.
\end{lemma}
\begin{proof}
	Let $A$ be an $f$-dual-approximation algorithm for $\Pi$. We proceed to construct a polynomial-time approximation algorithm for $\Pi$ with the claimed ratio. We consider the two statements of the lemma separately.
	

        \paragraph*{Case 1: $f(k) = \Ocal(k^c)$.}
        Given an $n$-vertex graph $G$ as instance of $\Pi$, we find $k_0 \in \{0,\dots,n\}$ defined as the largest positive integer $k$ such that algorithm~$A$ returns that $\opt(G) \geq k$.
	Note that $k_0$ can be found in polynomial time by performing at most $n+1$ calls to algorithm~$A$. If there is no such $k_0$, or if $k_0=0$, then $\opt(G) < \max(f(0),f(1)) = \Ocal(1)$, and since the decision version of $\Pi$ is in \NP, we can find an optimal solution in polynomial time by verifying all vertex subsets of size at most $\max(f(0),f(1))$.
        Otherwise, that is, if $k_0 \geq 1$, our approximation algorithm returns $k_0$, or if is constructive it returns a solution $S_0$ (a subset of vertices here) such that $|S_0| \ge k_0$.
	Let us prove that it provides the claimed approximation ratio.
        We distinguish two subcases depending on the value of $k_0$.
 	Suppose first that $k_0 \geq n^{1/c}$.
 	Since $\opt(G) \leq n$, in this case we get that
	\begin{equation*}\label{eq:approx1}
	\frac{\opt(G)}{k_0}\ \leq\ \frac{n}{n^{1/c}}\ =\ n^{\frac{c-1}{c}}.
	\end{equation*}

 	Otherwise, it holds that $k_0 < n^{1/c}$.
 	By the definition of $k_0$ we have $\opt(G) < f(k_0 +1) = \Ocal ( (k_0 +1)^{c}) = \Ocal ( (k_0) ^{c})$. Thus, in this case we get that
	 \begin{equation*}\label{eq:approx2}
	\frac{\opt(G)}{k_0}\ =\ \frac{\Ocal((k_0) ^{c})}{k_0}\ = \ \Ocal\left( (k_0) ^{c-1}\right)\ = \ \Ocal\left(n^{\frac{c-1}{c}}\right).
	\end{equation*}
	 Since in both cases we have a ratio of $\Ocal(n^{\frac{c-1}{c}})$, the lemma follows in Case 1.

         \paragraph*{Case 2: $f(k) = \beta k+1$.}
         Let $\varepsilon > 0$ be a arbitrary real number, let $\varepsilon' = \frac{\varepsilon}{\beta}$, and let us provide a polynomial-time approximation algorithm with ratio  $\beta(1+\varepsilon')=\beta+\varepsilon$. 
         As in Case 1, we start by finding $k_0$, defined as the largest positive integer $k$ such that algorithm~$A$ returns that $\opt(G) \geq k$.
         By definition of $k_0$ we have $\opt(G) < f(k_0 +1) = \beta(k_0+1)+1 \le \beta(k_0+2)$.
         If $k_0 < \frac{2}{\varepsilon'}$, then $\opt(G)$ is constant, and again by enumerating all subsets of size at most $\beta(\frac{2}{\varepsilon'}+2)$ we find an optimal solution.
         Otherwise, we return $k_0$, or if is constructive it returns a solution $S_0$ (a subset of vertices here) such that $|S_0| \ge k_0$.
         We have $\opt(G) \le \beta(k_0+2) \le \beta(1+\varepsilon')k_0$, concluding Case 2 of the proof.
\end{proof}

As a vertex-maximization problem whose decision version is in \NP is a well-behaved problem, the hypothesis of \autoref{lem:lop-dual} is satisfied (taking the identity function as upper bound function),
and thus the following theorem is immediate by pipelining \autoref{lem:lop-dual} and \autoref{lem:dual-approx}.

\begin{theorem}\label{thm:lop-kernels}
  Let  $\Pi$ be a vertex-maximization problem whose decision version is in \NP.
  \begin{enumerate}
  \item For every real number $c> 1$, if $\Pi$  admits a \lop-kernel with $\Ocal(k^c)$ vertices, then it admits a polynomial-time value-approximation algorithm with ratio $\O(n^{\frac{c-1}{c}})$ on $n$-vertex graphs.
  \item For every real number $\beta \ge 1$, if $\Pi$  admits a \lop-kernel with $\beta k$ vertices, then for any real number $\varepsilon > 0$, it admits a polynomial-time value-approximation algorithm with ratio $(\beta+1+\varepsilon)$.
  \end{enumerate}
  Moreover, if the \lop-kernel is constructive, then the corresponding approximation algorithm is also constructive.
\end{theorem}



As the framework of \lop-kernels is mainly defined as a tool to get \lop-kernel lower bounds from inapproximability, let us explicitly formulate the contrapositive of \autoref{thm:lop-kernels}.
Note that, when applying it to a concrete problem $\Pi$, the inapproximability of $\Pi$ will rely on some complexity assumption, typically $\P \neq\NP$.
\begin{corollary}\label{cor:lop-kernels}
  Let  $\Pi$ be a vertex-maximization problem whose decision version is in \NP.
  \begin{enumerate}
  \item For every real number $r \in (0,1)$, if $\Pi$ does not admit a polynomial-time value-approximation algorithm with ratio $\O(n^{r})$ on $n$-vertex graphs, then
    $\Pi$ parameterized by the solution size does not admit a \lop-kernel with $\Ocal(k^{\frac{1}{1-r}})$ vertices.
  \item For every real number $\beta > 1$, if $\Pi$ does not admit a polynomial-time value-approximation algorithm with ratio $\beta$, then
     $\Pi$ parameterized by the solution size does not admit a \lop-kernel with $(\beta-1-\varepsilon)k$ vertices for any real number $\varepsilon > 0$.
  \end{enumerate}
  Moreover, if the non-existence of approximation algorithms only holds for constructive approximation algorithms, then the lower bound
  only holds for constructive \lop-kernels.
\end{corollary}

\subsection{Connection between \lop-kernels and approximation algorithms for well-behaved maximization problems}
\label{sec:res-max-gnl}

The following lemma and theorem are the versions of \autoref{lem:dual-approx} and \autoref{thm:lop-kernels}, respectively, in the more general setting of arbitrary well-behaved maximization problems.

%

\begin{lemma}\label{lem:dual-approxG}
  Let $\Pi$ be a well-behaved maximization problem, $a \in \mathds{R}^+$, $\ub: \mathds{N} \to \mathds{N}$, and $f: \mathds{N} \to \mathds{N}$ be functions such that $\ub(n) = \O(n^a)$,
    $\ub$ is polynomial-time computable,  and $f$ is computable. Suppose that $\Pi$ has $\ub$ as upper bound function and that it admits an $f$-dual-approximation.
        \begin{enumerate}
          \item If $f(k) = \O(k^d)$ for some real number $d > 1$, then  $\Pi$ admits a polynomial-time value-approximation algorithm with ratio $\O(n^{\frac{a(d-1)}{d}})$, where $n$ is the size of the input.
	  \item If $f(k)=\lambda k^d+k+1$ for some real numbers $d \le 1$ and $\lambda > 0$, then $\Pi$ admits a polynomial-time value-approximation algorithm with ratio
            $\lambda 2^{d} + 3$.
        \end{enumerate}
	Moreover, if the $f$-dual-approximation algorithm is constructive, then the corresponding approximation algorithm is also constructive.
\end{lemma}
\begin{proof}

	Let $A$ be an $f$-dual-approximation algorithm for $\Pi$. For both cases in the statement in the lemma, we proceed to construct a polynomial-time approximation algorithm for $\Pi$ with the claimed ratio.
	


Given an instance $I$ of $\Pi$, we find $k_0 \in \{0,\dots,\ub(n)\}$ (recall that $n=|I|$)
  defined as the largest positive integer $k$ such that algorithm~$A$ returns that $\opt(G) \geq k$.
	Note that $k_0$ can be found in polynomial time as $|I|$ is polynomial-time computable and its value $n$ is polynomially upper-bounded in the classical bit-size of the instance, and that $\ub(n)$ can be computed in polynomial time as well. If there is no such $k_0$, or if $k_0=0$, then $\opt(G) < \max(f(0),f(1)) = \Ocal(1)$, and since
        $\Pi$ is well-behaved, we can, given $\max(f(0),f(1))$, decide in polynomial time if either $\opt(I)=0$ and provide a solution $s$ with $\val(I,s)=0$, or provide a solution $s$ with $\val(I,s) > 0$.
        In both cases we even have a constructive constant-factor approximation, and as $\opt(I) < f(1)=\lambda+2$, we get the claimed ratio in both cases.
        Otherwise, that is, when $k_0 \geq 1$, our approximation algorithm returns $k_0$, or if is constructive it returns a solution $s_0$ such that $\val(I,s_0) \ge k_0$.
	Let us prove that it provides the claimed approximation ratio.
        We now distinguish the two cases claimed in the statement of the lemma.

%

        \paragraph*{Case 1: $f(k) = \O(k^d)$.}
        Suppose first that $k_0 \geq n^{a/d}$.
        In this case we have
	\begin{equation*}
	\frac{\opt(I)}{k_0}\ \le\ \frac{\ub(n)}{k_0}\ =\ \frac{\O(n^a)}{n^{a/d}}\ =\ \O\left(n^{\frac{a(d-1)}{d}}\right).
	\end{equation*}
	Otherwise, it holds that $k_0 < n^{a/d}$.
        By the definition of $k_0$ we have $\opt(G) < f(k_0 +1)  = \Oh((k_0 + 1)^{d})= \Oh((k_0)^{d})$.
        Thus, in this case we get that
	\begin{equation*}
	\frac{\opt(I)}{k_0}\ =\ \frac{\Oh((k_0)^{d})}{k_0}\ = \ \Oh\left( (k_0) ^{d-1}\right)\ = \ \O\left(n^{\frac{a(d-1)}{d}}\right).
	\end{equation*}
         Since in both cases we have a ratio of $\Ocal(n^{\frac{a(d-1)}{d}})$, the lemma follows in Case 1.

         \paragraph*{Case 2: $f(k)=\lambda k^d+k+1$.}

         We have
         \begin{equation*}
	   \frac{\opt(I)}{k_0}\ \le \ \frac{f(k_0 + 1)}{k_0}\ \le \ \frac{\lambda (k_0 + 1)^{d} + k_0+2 }{k_0}
	 \end{equation*}
	
	Let
		\begin{equation*}
		h(x) = \frac{\lambda (x + 1)^{d} + x + 2}{x}
		\end{equation*}
	and note that the approximation ratio is at most $h(k_0)$.
	
	To bound $h(k_0)$, we proceed to show that $h(x)$ is decreasing in $x$ when $x > 0$ and obtain the desired approximation ratio of $h(k_0) \le h(1) = \lambda 2^{d} + 3$.
	To show that $h(x)$ is indeed decreasing in $x$ when $x > 0$, note that
	\begin{equation*}\label{eq:somethingsomethingdarkside}
	\frac{\partial h(x)}{\partial x} = \frac{\lambda (x + 1)^{d}}{x^2} \cdot \left(d\frac{x}{x + 1}-1\right) - \frac{1}{x^2},
	\end{equation*}
	which is negative when $d \le 1$ and $x > 0$.
\end{proof}

The next theorem follows immediately by pipelining \autoref{lem:lop-dual} and \autoref{lem:dual-approxG}.  Namely, starting with the hypothesis of \autoref{thm:lop-kernelsG}, we first apply \autoref{lem:lop-dual} and then \autoref{lem:dual-approxG} with $d = ac$ and $\lambda = \alpha \beta^a$.


\begin{theorem}\label{thm:lop-kernelsG}
  Let  $\Pi$ be a well-behaved maximization problem, $a, c \in \mathds{R}^+$, $\ub: \mathds{N} \to \mathds{N}$, and $s: \mathds{N} \to \mathds{N}$ be functions such  that $\ub(n) = \O(n^a)$, $s(k) = \O(k^c)$, $\ub$ is non-decreasing, and $s$ and $\ub$ are polynomial-time computable.
Suppose that $\Pi$ has $\ub$  as upper bound function and that it admits a \lop-kernel of size $s$, according to the same size function $|\cdot|$ associated with $\Pi$.
\begin{enumerate}
\item If $ac > 1$, then  $\Pi$ admits a polynomial-time value-approximation algorithm with ratio $\O(n^{\frac{ac-1}{c}})$,
where $n$ is the size of the instance.
 \item If $ac \le 1$ and $\alpha, \beta\in \mathds{R}^+$ are such that $\ub(n) \le \alpha n^a$ and $s(k) \le \beta k^c$, then $\Pi$ admits a polynomial-time value-approximation algorithm with ratio
   $\alpha \beta^a 2^{ac} + 3$.
\end{enumerate}
Moreover, if the \lop-kernel is constructive, then the corresponding approximation algorithm is also constructive.
\end{theorem}

To provide some insight on the formulas used in the statement of \autoref{thm:lop-kernelsG}, and especially on the role of the upper bound function $\ub(n) = \O(n^a)$, one can typically think of a graph problem where the output is a subset of edges, and where the size of an instance is the number of vertices of input graph.
In that case, we  have $a=2$, and thus a \lop-kernel of size (in terms of number of vertices) $\O(k^c)$ would only imply an $\O(n^{\frac{2c-1}{c}})$-approximation algorithm, which is worse
than the ratio $\O(n^{\frac{c-1}{c}})$ obtained in \autoref{thm:lop-kernels}, where $a=1$.
On the other hand, the ratio can also sometimes be slightly better, as there may exist problems with upper bound function $\ub(n) = \O(n^a)$ for some $a <1$.
Moreover, observe that for problems where $a\le 1$, the second item covers the case of linear kernels, which corresponds to $c=1$.


By taking the contrapositive of \autoref{thm:lop-kernelsG}, we obtain the following more general version of \autoref{cor:lop-kernels}.


\begin{corollary}\label{cor:lop-kernelsG}
  Let  $\Pi$ be a well-behaved maximization problem with a non-decreasing and polynomial-time computable upper bound function $\ub(n) = \O(n^a)$ for $a \in \mathds{R}^+$.
  In what follows, the size of the instance, denoted by $n$, and the size of the kernel are defined according to the same size function $|\cdot|$ associated with~$\Pi$.
  \begin{enumerate}
  \item For every real number $r \in (0,1)$, if $\Pi$ does not admit a polynomial-time value-approximation algorithm with ratio $\O(n^{r})$, then
    $\Pi$ parameterized by the solution size does not admit a \lop-kernel of size $\Ocal(k^{\frac{1}{a-r}})$.
  \item Suppose that $\ub(n) \le \alpha n^a$ for some $\alpha \in \mathds{R}^+$.
    For every real number $\beta > 1$, if $\Pi$ does not admit a polynomial-time value-approximation algorithm with ratio $\beta$, then
    $\Pi$ parameterized by the solution size does not admit a \lop-kernel of size $\beta' k^{c'}$ for any real numbers $\beta'$, $c'$ such that
    $ac' \le 1$ and $\alpha \beta^{'a}2^{ac'}+3 \le \beta$.
  \end{enumerate}
  Moreover, if the non-existence of approximation algorithms only holds for constructive approximation algorithms, then the lower bound
  only holds for constructive \lop-kernels.
\end{corollary}

\section{A framework for ruling out certain polynomial kernels: the case of minimization problems}
\label{sec:framework-min}

In this section we adapt the framework of \lop-kernels introduced in \autoref{sec:framework-max}
to minimization problems. The definitions and results for minimization problems are very close to those for maximization problems,  but there are a number of subtle differences that we will discuss as they appear.

We say that a minimization problem $\Pi$ is \emph{well-behaved} if it comes equipped with a size function (as defined in \autoref{sec:prelim}) and it satisfies the following condition, which we denote by $C^{\min}$:

\begin{quotation}
  \noindent There exists a polynomial-time algorithm that, given as input  an instance $I$ of $\Pi$, decides if $\opt(I)=0$, and in this case provides a solution $s$ where $\val(I,s)=0$, or
otherwise provides any solution $s$.
\end{quotation}

Note that condition $C^{\min}$ above is {\sl not} the symmetric version of condition $C^{\max}$ defined at the beginning of \autoref{sec:framework-max}.



Given a well-behaved minimization problem $\Pi$,
we say that a function $\ub: \mathds{N} \to \mathds{N}$ is an \emph{upper bound function} for $\Pi$ if for any instance $I$ of $\Pi$ and any solution $s$ of $I$, we have $\val(I,s) \le \ub(|I|)$, where $|\cdot|$ is the size function of $\Pi$. Note that this notion of upper bound function differs from the one given for maximization problems.
Again, throughout the paper, we assume that the notions of size used in both the size of kernels and the upper bound function are the same.

The reminder of this section is organized similarly to \autoref{sec:framework-max}. Namely, in \autoref{sec:lop-defs-for-min} we present the definition of \lop-rules and \lop-kernels for well-behaved minimization problems, and we prove a general technical result, namely \autoref{lem:lop-dualmin}.
In \autoref{sec:res-min} we present the connection between \lop-kernels and approximation algorithms for vertex-minimization problems, and in \autoref{sec:res-min-gnl}
 we generalize it to arbitrary well-behaved minimization problems.

\subsection{Definition of \lop-rules and \lop-kernels}
\label{sec:lop-defs-for-min}

The following definition should be compared to \autoref{def:lop}.

\begin{definition}\label{def:lop-min}
  A \emph{large optimal preserving} reduction rule, or \emph{\lop-rule} for short, for a well-behaved minimization problem $\Pi$, is a polynomial-time algorithm $R$ that, given an instance
  $(I,k)$ of $\Pi_{\dec}$, computes another instance $(I',k')$ of $\Pi_{\dec}$ with $0 \leq k' \leq k$ and such that
  \begin{enumerate}
    \item if $(I,k)$ is a \yes-instance of $\Pi_{\dec}$, then $(I',k')$ is a \yes-instance of $\Pi_{\dec}$, and
    \item\label{point3min}  if $(I,k)$ is a  \no-instance of $\Pi_{\dec}$, then $\opt(G') \geq \opt(G) - (k-k')$, implying that $(I',k')$ is a \no-instance of $\Pi_{\dec}$.
  \end{enumerate}
  A \lop-rule $R$ is \emph{constructive} if, for any solution $s'$ of $I'$, it constructs (in polynomial time) a solution $s$ of $I$ such that $\val(I,s) \le \val(I',s')+(k-k')$.
\end{definition}

Note that Property~\ref{point3min} in \autoref{def:lop-min} is stronger than the implication ``if $(I,k)$ is a \no-instance  of $\Pi_{\dec}$, then $(I',k')$ is a \no-instance  of $\Pi_{\dec}$'', which would yield the definition of a classical kernelization algorithm.
Observe also that the constructibility condition implies Property~\ref{point3min}, unlike in the maximization case. 
Indeed, when we consider how this latter implication is generally proved in safeness proofs of classical kernels, we generally prove its contrapositive, and one of the following scenarios often occur:
\begin{enumerate}[(a)]
\item\label{case1min} For every solution $s'$ of $I'$ there exists a solution $s$ of $I$ with $\val(I,s) \le \val(I',s')+(k-k')$.
\item\label{case2min} For every solution $s'$ of $I'$ with $\val(I',s') \le k'$, there exists a solution $s$ of $I$ with $\val(I,s) \le \val(I',s')+(k-k')$.
\item\label{case3min} If there exists a solution $s'$ of $I'$ with $\val(I',s') \le k'$, then there exists a solution $s$ of $I$ with $\val(I,s) \le  k$.
\end{enumerate}

In Case~(\ref{case1min}), the rule preserves all optimal values, as it implies that $\opt(G') \ge \opt(G)-(k-k')$, and note that it implies Property~\ref{point3min} in \autoref{def:lop-min}. Compared to the maximization case, Case~(\ref{case3min}) still implies only the classical implication ``if $(I,k)$ is a \no-instance of $\Pi_{\dec}$, then $(I',k')$ is a \no-instance of $\Pi_{\dec}$'', but note that Case~(\ref{case2min}) no longer implies Property~\ref{point3min}. This is one of reasons why, in our opinion, the framework of \lop-kernels seems to be more natural when applied to maximization problems.


The following observation is again an immediate consequence of the definition of a \lop-rule.

\begin{observation}\label{prop:lop-chain-min}
  \lop-rules can be composed. Formally, consider two \lop-rules $R_1$ and $R_2$.
  Then, the rule $R$ that, given a instance $(I,k)$ of $\Pi_{\dec}$, returns $R_2(R_1(I,k))$, is also a \lop-rule. Moreover, if $R_2$ and $R_1$ are constructive, then $R$ is also constructive.
\end{observation}


A typical example of a \lop-rule for a vertex-minimization problem is when, for some problem $\Pi$, we can identify a ``dominant'' set of vertices that can be safely included into a solution. More precisely, consider a rule that, given a graph $G$, finds a subset $T \subseteq V(G)$ and a graph $G'$ such that there exists an optimal solution $S^{\star}$ in $G$ such that $S^{\star} = T \cup S'$, where $S'$ is a solution in $G'$, and  for every solution $S'$ in $G'$, $S' \cup T$ is a solution in $G$.
Such a rule is indeed a (constructive) \lop-rule.

Even if almost all classical known reduction rules for minimization problems are \lop-rules~\cite{Book-kernels,CyganFKLMPPS15}, here is a simple example a non-\lop-rule.
Consider the \textsc{Vertex Cover} problem, and suppose that, given an instance $(G,k)$, we find in $G$ a matching $M$ of size $k+1$.
The rule just outputs $(G',k')=(M,k)$, hence preserving the fact that $(G,k)$ is a \no-instance.
However, this rule does not satisfy Property~\ref{point3min} in \autoref{def:lop-min},
since the size of a minimum vertex cover of $G$
may be arbitrarily large compared to $k$, hence the inequality  $\opt(G') \geq \opt(G)$ may not hold. In \autoref{sec:applications-lop} we discuss a (much more involved) reduction rule for a vertex-minimization problem, namely \textsc{Tree Deletion Set}, which is not a \lop-rule either.

As in the maximization case, if we defined a \lop-kernel as an  algorithm consisting only of \lop-rules, we would exclude from being a \lop-kernel, for instance, the algorithm consisting of the rule that detects a \no-instance of \textsc{Vertex Cover} as in the above paragraph. This justifies the next definition, where we also allow \lop-kernels to decide instances, and that should be compared to \autoref{def:lop-kernel}.



\begin{definition}\label{def:lop-kernelmin}
 Let $\Pi$ be a well-behaved minimization problem and let $s: \mathds{N} \to  \mathds{N}$ be a computable function. A \emph{\lop-kernel of size $s$} for $\Pi$ parameterized by the solution size  is
 a polynomial-time algorithm that takes as input an instance $(I,k)$ of $\Pi_{\dec}$, and either
  \begin{itemize}
    \item decides that $(I,k)$ is a \yes-instance or a \no-instance, or
    \item outputs a reduced instance $(I',k')$ by applying a sequence of \lop-rules to $(I,k)$, with $|I'| \leq s(k)$.
  \end{itemize}
A \lop-kernel is \emph{constructive} if, in the first case, it constructively decides $(I,k)$, and, in the second case, it only uses constructive \lop-rules.
\end{definition}

Note that the constructibility condition in \autoref{def:lop-kernelmin} differs from that in \autoref{def:lop-kernel} for maximization problems. We need this stronger property in the proof of \autoref{lem:lop-dualmin}.





Our next objective is to prove that a \lop-kernel yields the existence of a polynomial-time approximation algorithm.
For this, we need the following definition, which is the version of \autoref{def:dual-approx} for minimization problem.

 \begin{definition}\label{def:dual-approxmin}
   Let $\Pi$ be a well-behaved minimization problem and let $f: \mathds{N} \to \mathds{N}$. An \emph{$f$-dual-approximation algorithm for $\Pi$}  is a polynomial-time algorithm that, given an instance $(I,k)$ of $\Pi_{\dec}$, concludes one of the following:
   \begin{itemize}
   \item $\opt(I) \leq f(k)$.
   \item $\opt(I) > k$.
   \end{itemize}
 An $f$-dual-approximation algorithm is \emph{constructive} if, whenever it concludes that $\opt(I) \leq f(k)$, it provides a solution  $s$ with $\val(I, s) \le f(k)$.
 \end{definition}

 In the next lemma we prove that a \lop-kernel of size $s$ yields an $f$-dual-approximation algorithm (where $f$ depends on $s$), which in turn yields a classical approximation algorithm whose ratio depends on $s$.
 As in the maximization case, to provide some insight on the statement of the next lemma, keep in mind that for vertex-minimization problems, the upper bound function $u$ is typically the identity function. Note that in the next lemma, the derived function $f(k)$ differs slightly from that of \autoref{lem:lop-dual}; this is due to technical reasons motivated by the fact that there the maximization and minimization versions of our framework are not totally symmetric.


 \begin{lemma}\label{lem:lop-dualmin}
	Let $\Pi$ be a well-behaved minimization problem with a non-decreasing upper bound function $\ub$ and let $s: \mathds{N} \to \mathds{N}$ be a computable function.
	If $\Pi$ admits a \lop-kernel of size $s$, then $\Pi$ admits an $f$-dual-approximation algorithm where $f(k) := \ub(s(k)) + k$.
	Moreover, if the \lop-kernel is constructive, then the $f$-dual-approximation algorithm is also constructive.
\end{lemma}
\begin{proof}
	Let $k\in \mathds{N}$, $(I,k)$ be an instance of $\Pi_{\dec}$, and $\Rcal$ be a \lop-kernel of size $s$ for $\Pi$.
	We describe an $f$-dual-approximation algorithm $\Acal$ which takes as input $I$ and $k$, starts by running $\Rcal$ with input $(I, k)$, and continues based on its possible output.
        If $\Rcal$ decides that $(I,k)$ is a \no-instance, then $\opt(I)> k$, and $\Acal$ returns $\opt(I)> k $.
        If $\Rcal$ decides that $(I,k)$ is a \yes-instance, then $\opt(I)\le k$, and $\Acal$ returns $\opt(I)\le k \le f(k)$ as well.
        Notice that if $\Rcal$ is constructive, then it provides a solution $s$ of $I$ such that $\val(I,s) \ge k$, and $\Acal$ returns this solution as well.

        Finally, suppose that $\Rcal$ outputs an equivalent instance $(I',k')$ obtained from $(I, k)$ using only \lop-rules and such that $|I'| \leq s(k)$.
        By using \autoref{prop:lop-chain-min} we can assume that $(I',k')$ is obtained from $(I,k)$ by a single \lop-rule.
        Let us start by the non-constructive case, in which $\Acal$ returns $\opt(I) \le f(k)$.
        If $\opt(I) \le k$, then we are done as $k \le f(k)$.
        If $\opt(I) > k$, Property~\ref{point3min} in \autoref{def:lop-min} implies that $\opt(I) \le \opt(I')+(k-k') \le \ub(|I'|) + k  \leq   \ub(s(k)) + k = f(k)$, where we have used that $\ub$ is non-decreasing.
	Let us now turn to the constructive case. As $\Pi$ is well-behaved and verifies $C^{\min}$, we can compute in polynomial time a solution $s'$ (of any cost), and according to the definition
        of $\ub$ we have $\val(I',s') \le \ub(|I'|) \le \ub(s(k))$, where we have used again that $\ub$ is non-decreasing. Finally, as the rule is constructive, we can construct in polynomial time a solution $s$ such that $\val(I,s) \le \val(I',s')+k \le f(k)$,
        and algorithm $\Acal$ returns this solution as well.
\end{proof}



Let us now turn to our main results relating the size of \lop-kernels with the existence of approximation algorithms. As in \autoref{sec:framework-max}, to keep statements as simple as possible, we provide in \autoref{sec:res-min} results that correspond to the specialized versions for vertex-maximization problems of results in \autoref{sec:res-min-gnl}.

\subsection{Connection between \lop-kernels and approximation algorithms for vertex-minimization problems}
\label{sec:res-min}

In this subsection we deal with vertex-minimization problems. The following lemma, which should be compared to \autoref{lem:dual-approx}
Note that the hypothesis in the second item of \autoref{lem:dual-approx} is slightly different from the one below, and that the obtained approximation ratios are also slightly different.


\begin{lemma}\label{lem:dual-approx-min}
  Let $\Pi$ be a vertex-minimization problem whose decision version is in \NP, $c > 1$ and $\beta \ge 1$ be real numbers, and $f: \mathds{N} \to \mathds{N}$ be a computable function.
  \begin{enumerate}
  \item If $\Pi$ admits an $f$-dual-approximation algorithm where $f(k) = \Ocal(k^c)$, then $\Pi$ admits a polynomial-time value-approximation algorithm with ratio $\O(n^{\frac{c-1}{c}})$ on $n$-vertex graphs.
  \item If $\Pi$ admits a $f$-dual-approximation algorithm where $f(k) = \beta k$, then $\Pi$ admits a polynomial-time value-approximation algorithm with ratio $\beta$.
  \end{enumerate}
  Moreover, if the $f$-dual-approximation is constructive, then the corresponding approximation algorithm is also constructive.
\end{lemma}
\begin{proof}
	Let $A$ be an $f$-dual-approximation algorithm for $\Pi$. We proceed to construct a polynomial-time approximation algorithm for $\Pi$ with the claimed ratio. We consider the two statements of the lemma separately.
	

        \paragraph*{Case 1: $f(k) = \Ocal(k^c)$.}
        Given an $n$-vertex graph $G$ as instance of $\Pi$, we find $k_0 \in \{0,\dots,n\}$ defined as the smallest positive integer $k$ such that algorithm~$A$ returns that $\opt(G) \le f(k)$.
	Note that $k_0$ can be found in polynomial time by performing at most $n+1$ calls to algorithm~$A$.
        Notice that $k_0$ always exists as we cannot have $\opt(G) > n$.
        If $k_0=0$, then $\opt(G) \le f(0) = \Ocal(1)$, and since the decision version of $\Pi$ is in \NP, we can find an optimal solution in polynomial time by verifying all vertex subsets of size at most $f(0)$. Otherwise, that is, if $k_0 \ge 1$, our approximation algorithm returns $f(k_0)$, or if is constructive it returns a solution $S_0$ (that is, a subset of vertices) such that $|S_0| \le f(k_0)$. By definition of $k_0$, we have that $\opt(G) > k_0 -1$, or equivalently $\opt(G) \ge k_0$.
	Let us prove that this algorithm provides the claimed approximation ratio.
        We distinguish two subcases depending on the value of $k_0$.
 	Suppose first that $k_0 \geq n^{1/c}$.
 	In this case we get that
	\begin{equation*}\label{eq:approx1min}
	\frac{f(k_0)}{\opt(G)}\ \le\ \frac{n}{k_0}\ \le \ \Ocal(n^{\frac{c-1}{c}}).
	\end{equation*}

 	Otherwise, it holds that $k_0 < n^{1/c}$.
 	In this case we get that
	 \begin{equation*}\label{eq:approx2min}
	\frac{f(k_0)}{\opt(G)}\ \le \ \frac{f(k_0)}{k_0}\ <\  \Ocal(k_0^{c-1})\ =\ \Ocal(n^{\frac{c-1}{c}}).
	\end{equation*}
	 Since in both cases we have a ratio of $\Ocal(n^{\frac{c-1}{c}})$, the lemma follows in Case 1.

         \paragraph*{Case 2: $f(k) = \beta k$.}
         As in Case 1, we start by finding $k_0$ defined as the smallest positive integer $k$ such that algorithm~$A$ returns that $\opt(G) \le f(k)$.
         If $k_0=0$ we proceed as in the first case.
         Otherwise, we return $f(k_0)$, or if the algorithm is constructive we return a solution $S_0$ (that is, a subset of vertices) such that $|S_0| \ge f(k_0)$.
         We have
         \begin{equation*}\label{eq:approx3min}
	\frac{f(k_0)}{\opt(G)}\ \le \ \frac{f(k_0)}{k_0}\ \le\ \beta,
	\end{equation*}
and the lemma follows in Case 2.
\end{proof}

As a vertex-minimization problem whose decision version is in \NP is a well-behaved problem, the hypothesis of \autoref{lem:dual-approx-min} is satisfied (by taking the identity function as upper bound function),
and thus the following theorem is immediate by pipelining \autoref{lem:lop-dualmin} and \autoref{lem:dual-approx-min}. The next theorem should be compared to \autoref{thm:lop-kernels}.

\begin{theorem}\label{thm:lop-kernelsmin}
  Let  $\Pi$ be a vertex-minimization problem whose decision version is in \NP.
  \begin{enumerate}
  \item For every real number $c > 1$, if $\Pi$  admits a \lop-kernel with $\Ocal(k^c)$ vertices, then it admits a polynomial-time value-approximation algorithm with ratio $\O(n^{\frac{c-1}{c}})$ on $n$-vertex graphs.
  \item For every real number $c > 1$, if $\Pi$  admits a \lop-kernel with at most $ck$ vertices, then it admits a polynomial-time value-approximation algorithm with ratio $(c+1)$.
  \end{enumerate}
  Moreover, if the \lop-kernel is constructive, then the corresponding approximation algorithm is also constructive.
\end{theorem}




As the framework of \lop-kernels is mainly defined as a tool to get \lop-kernel lower bounds from inapproximability, let us explicitly formulate the contrapositive of \autoref{thm:lop-kernelsmin}.
Note again that, when applying it to a concrete problem $\Pi$, the inapproximability of $\Pi$ will rely on some complexity assumption, typically $\P \neq\NP$.

\begin{corollary}\label{cor:lop-kernelsmin}
  Let  $\Pi$ be a vertex-minimization problem whose decision version is in \NP.
  \begin{enumerate}
  \item For every real number $r \in (0,1)$, if $\Pi$ does not admit a polynomial-time value-approximation algorithm with ratio $\O(n^{r})$ on $n$-vertex graphs, then
    $\Pi$ parameterized by the solution size does not admit a \lop-kernel with $\Ocal(k^{\frac{1}{1-r}})$ vertices.
  \item For every real number $\beta > 1$, if $\Pi$ does not admit a polynomial-time value-approximation algorithm with ratio $\beta$, then
    $\Pi$ parameterized by the solution size does not admit a \lop-kernel with $(\beta-1-\varepsilon)k$ vertices  for any real number $\varepsilon > 0$.
  \end{enumerate}
  Moreover, if the non-existence of approximation algorithms only holds for constructive approximation algorithms, then the lower bound
  only holds for constructive \lop-kernels.
\end{corollary}

\subsection{Connection between \lop-kernels and approximation algorithms for well-behaved minimization problems}
\label{sec:res-min-gnl}

The following lemma and theorem are a more general version of  \autoref{lem:dual-approx-min} and \autoref{thm:lop-kernelsmin}, respectively, for well-behaved minimization problems. The next lemma should be compared to \autoref{lem:dual-approxG}, and note that the obtained approximation ratios in the second item of both lemmas are different.

\begin{lemma}\label{lem:dual-approxGmin}
  Let $\Pi$ be a well-behaved minimization problem, $a \in \mathds{R}^+$, $\ub: \mathds{N} \to \mathds{N}$, and $f: \mathds{N} \to \mathds{N}$ be functions such that $\ub(n) = \O(n^a)$,
    $\ub$ is polynomial-time computable,  and $f$ is computable. Suppose that $\Pi$ has \ub as upper bound function $\ub$ and that it admits an $f$-dual-approximation.
        \begin{enumerate}
          \item If $f(k) = \O(k^d)$ for some real number $d > 1$, then  $\Pi$ admits a polynomial-time value-approximation algorithm with ratio $\O(n^{\frac{a(d-1)}{d}})$, where $n$ is the size of the input.
	  \item If $f(k)=\lambda k^d+k$ for some real numbers $d \le 1$ and $\lambda > 0$, then $\Pi$ admits a polynomial-time value-approximation algorithm with ratio $\lambda + 1$.
        \end{enumerate}
	Moreover, if the dual-approximation algorithm is constructive, then the corresponding approximation algorithm is also constructive.
\end{lemma}
\begin{proof}

	Let $A$ be an $f$-dual-approximation algorithm for $\Pi$. For both cases in the statement of the lemma, we proceed to construct a polynomial-time approximation algorithm for $\Pi$ with the claimed ratio.


        Given an instance $I$ of $\Pi$, we find $k_0 \in \{0,\dots,\ub(n)\}$ (recall that $n=|I|$) defined as the smallest positive integer $k$ such that algorithm~$A$ returns that $\opt(G) \leq f(k)$.
	Note that $k_0$ can be found in polynomial time, as $|I|$ is polynomial-time computable and its value $n$ is polynomially upper-bounded in the classical bit-size of the instance, and that $\ub(n)$ can be computed in polynomial time as well. Note that $k_0$ always exists, as we cannot have $\opt(I) > \ub(n)$.
        If $k_0=0$, then as by hypothesis $\Pi$ satisfies property $C^{\min}$,
         we can verify in polynomial time whether $\opt(I)=0$ and, if it is the case, we provide an optimal solution $s$ with $\val(I,s)=0$.
        Otherwise, we have $\opt(I) \ge 1$, and $A$ returns $f(k_0)$, or a solution $s$ such that $\val(I,s) \le f(k_0)$ if the dual-approximation algorithm is constructive.
        In the first case (that is, if $f(k) = \O(k^d)$), as $\opt(I) \ge 1$, we have a ratio $f(0)$, implying the claimed ratio.
        In the second case (that is, if $f(k)=\lambda k^d+k$), $f(0)=0$, so we even have an optimal solution.

        Let us now assume that $k_0 \geq 1$, and recall that $\opt(I) \ge k_0$.
        We distinguish the two cases claimed in the statement of the lemma.


        \paragraph*{Case 1: $f(k) = \O(k^d)$.}
        Suppose first that $k_0 \geq n^{a/d}$.
        In this case we have
	\begin{equation*}
	\frac{f(k_0)}{\opt(I)}\ \le \ \frac{\ub(n)}{k_0}\ \le \ \frac{\O(n^a)}{k_0}\ = \ \O\left(n^{\frac{a(d-1)}{d}}\right).
	\end{equation*}
	Otherwise, it holds that $k_0 < n^{a/d}$
        In this case we get that
	\begin{equation*}
	\frac{f(k_0)}{\opt(I)}\ =\ \frac{\Oh((k_0)^{d})}{k_0}\ = \ \Oh\left( (k_0) ^{d-1}\right)\ = \ \O\left(n^{\frac{a(d-1)}{d}}\right).
	\end{equation*}
        Since in both cases we have a ratio of $\Ocal(n^{\frac{a(d-1)}{d}})$, the lemma follows in Case 1.

         \paragraph*{Case 2: $f(k)=\lambda k^d+k$.}
         In this case  we have
         \begin{equation*}
	   \frac{f(k_0)}{\opt(I)}\ \le \ \frac{\lambda k_0^d+k_0}{k_0}\ = \ \lambda (k_0)^{d-1} + 1.
	 \end{equation*}
	 As $d \le 1$, the last expression in the above equation is decreasing in $k_0$, and as $k_0 \ge 1$, the maximum is reached for $k_0=1$, and the approximation ratio claimed in Case 2 follows.
\end{proof}

The next theorem, which should be compared to \autoref{thm:lop-kernelsG}, follows immediately by pipelining \autoref{lem:lop-dualmin} and \autoref{lem:dual-approxGmin}. Namely, starting with the hypothesis of \autoref{thm:lop-kernelsGmin}, we first apply \autoref{lem:lop-dualmin} and then \autoref{lem:dual-approxGmin} with $d = ac$ and $\lambda = \alpha \beta^a$.

\begin{theorem}\label{thm:lop-kernelsGmin}
  Let  $\Pi$ be a well-behaved minimization problem, $a, c \in \mathds{R}^+$, $\ub: \mathds{N} \to \mathds{N}$, and $s: \mathds{N} \to \mathds{N}$ be functions such that $\ub(n) = \O(n^a)$, $s(k) = \O(k^c)$, $\ub$ is non-decreasing, and $s$ and $\ub$ are polynomial-time computable.
Suppose that $\Pi$ has $\ub$  as upper bound function and that it admits a \lop-kernel of size $s$, according to the same size function $|\cdot|$  associated with $\Pi$.
\begin{enumerate}
 \item If $ac > 1$, then  $\Pi$ admits a polynomial-time value-approximation algorithm with ratio $\O(n^{\frac{ac-1}{c}})$.
 \item If $ac \le 1$ and $\alpha, \beta\in \mathds{R}^+$ are such that $\ub(n) \le \alpha n^a$ and $s(k) \le \beta k^c$, then $\Pi$ admits a polynomial-time value-approximation algorithm with ratio $\alpha \beta^a + 1$.
\end{enumerate}
Moreover, if the \lop-kernel is constructive, then the corresponding approximation algorithm is also constructive.
\end{theorem}

The discussion provided right after \autoref{thm:lop-kernelsG} also applies to the above theorem. The contrapositive of \autoref{thm:lop-kernelsGmin} yields the following corollary.


\begin{corollary}\label{cor:lop-kernelsminG}
  Let  $\Pi$ be a well-behaved minimization problem with a non-decreasing polynomial-time computable upper bound function $\ub(n) = \O(n^a)$ for $a \in \mathds{R}^+$.
  In what follows, the size of the instance, denoted by $n$, and the size of the kernel are defined according to the same size function $|\cdot|$ associated with $\Pi$.
  \begin{enumerate}
  \item For every real number $r \in (0,1)$, if $\Pi$ does not admit a polynomial-time value-approximation algorithm with ratio $\O(n^{r})$, then
    $\Pi$ parameterized by the solution size does not admit a \lop-kernel of size $\Ocal(k^{\frac{1}{a-r}})$.
  \item Suppose that $\ub(n) \le \alpha n^a$ for some $\alpha \in \mathds{R}^+$.
    For every real number $\beta > 1$, if $\Pi$ does not admit a polynomial-time value-approximation algorithm with ratio $\beta$, then
    $\Pi$ parameterized by the solution size does not admit a \lop-kernel of size $\beta' k^{c'}$ for any real numbers $\beta'$, $c'$ such that
    $ac' \le 1$ and $\alpha \beta^{'a}+1 \le \beta$.
  \end{enumerate}
  Moreover, if the non-existence of approximation algorithms only holds for constructive approximation algorithms, then the lower bound
  only holds for constructive \lop-kernels.
\end{corollary}

\section{Applications of the framework of \lop-kernels}
\label{sec:applications-lop}


In this section we provide several applications of the framework of \lop-kernels introduced in \autoref{sec:framework-max} and \autoref{sec:framework-min}.

\medskip

Our first application concerns the {\sc Maximum Minimal Vertex Cover} problem, defined in \autoref{sec:prelim}. Boria et al.~\cite{BoriaCP15} proved that {\sc Maximum Minimal Vertex Cover} does not admit a polynomial-time $\O(n^{\frac{1}{2}-\varepsilon})$-approximation algorithm for any $\varepsilon>0$, unless $\P = \NP$. Hence,
by applying \autoref{cor:lop-kernels} with $r=\frac{1}{2} - \varepsilon$ we obtain the following corollary, which matches the best known kernel having $\Ocal(k^2)$ vertices~\cite{FernauHDR}.

\begin{corollary}\label{cor:lop1}
{\sc Maximum Minimal Vertex Cover} parameterized by the solution size does not admit a \lop-kernel with $\Ocal(k^{2 - \varepsilon})$ vertices for any $\varepsilon > 0$, unless $\P = \NP$.
\end{corollary}

Our second application is similar to the first one. In the {\sc Maximum Minimal Feedback Vertex Set} problem, given an $n$-vertex graph $G$ and an integer $k$, the objective is to decide if there exists a minimal feedback vertex set $S \subseteq V(G)$ (i.e., a set $S$ such that
$G\setminus S$ is a forest) of size at least $k$.
Dublois et al.~\cite{DubloisHGLM20} recently proved that the problem does not admit a polynomial-time $\O(n^{\frac{2}{3}-\varepsilon})$-approximation algorithm for any $\varepsilon>0$, unless $\P = \NP$. Hence, by applying \autoref{cor:lop-kernels} with $r=\frac{2}{3} - \varepsilon$ we obtain the following corollary, which matches the best known kernel with $\Ocal(k^3)$ vertices also provided by Dublois et al.~\cite{DubloisHGLM20}.

\begin{corollary}\label{cor:lop2}
{\sc Maximum Minimal Feedback Vertex Set} parameterized by the solution size does not admit a \lop-kernel with $\Ocal(k^{3 - \varepsilon})$ vertices for any $\varepsilon > 0$, unless $\P = \NP$.
\end{corollary}

Our third application concerns a vertex-minimization problem. In the {\sc Tree Deletion Set} problem, given a graph $G$ and an integer $k$, the objective is to decide whether at most $k$ vertices can be deleted from an $n$-vertex graph $G$ in order to obtain a tree.
It is known that this problem does not admit a polynomial-time $\O(n^{1-\varepsilon})$-approximation for any $\varepsilon > 0$ unless $\P \neq\NP$~\cite{Yannakakis79}.
\autoref{cor:lop-kernelsmin} implies the following.

\begin{corollary}\label{cor:lop3}
{\sc Tree Deletion Set} parameterized by the solution size does not admit a polynomial \lop-kernel, unless $\P = \NP$.
\end{corollary}

The interesting fact is that {\sc Tree Deletion Set} admits a kernel with $\O(k^4)$ vertices~\cite{GiannopoulouLSS16}.
This kernel is the only non-artificial example of \emph{non}-\lop-kernel that we are aware of so far. Thus, the algebraic reduction rule presented by Giannopoulou et al.~\cite{GiannopoulouLSS16}, which is based on identifying a subset of linear equations of appropriate size that captures all  solutions of size at most $k$, cannot be (even transformed to) a \lop-rule.

\medskip

Our last application deals with the {\sc Maximum Independent Set} problem restricted to $K_t$-free graphs, for an integer $t \geq 3$. Ramsey's theorem~\cite{Ramsey} implies that, given a $K_t$-free graph on $n$ vertices, it is always possible to find in polynomial time an independent set of size at least $n^{\frac{1}{t-1}}$.
This directly implies a polynomial-time $n^{\frac{t-2}{t-1}}$-approximation algorithm for {\sc Maximum Independent Set} on $K_t$-free graphs, and a constructive \lop-kernel of size $k^{t-1}$ (indeed, if the input graph has size al least $k^{t-1}$, we can safely declare it a \yes-instance).
To the best of our knowledge, improving this trivial approximation factor is still open, and the only known inapproximability result is the bound $\Ocal(n^{\frac{1}{4}-\varepsilon})$ on triangle-free graphs recently proved by Bonnet et al.~\cite{BonnetTTW20}, which relies on the hypothesis that $\NP \nsubseteq \BPP$. In the same paper~\cite{BonnetTTW20}, the authors state the following conjecture, called the ``Improved Approximation Conjecture'': for every fixed graph $H$, there exists a constant $\varepsilon > 0$ such that
{\sc Maximum Independent Set} admits a (randomized) polynomial-time $n^{1-\varepsilon}$-approximation algorithm on $H$-free $n$-vertex graphs. We state the following conjecture.

\begin{conjecture}\label{conj:lop}
For every fixed graph $H$, the {\sc Maximum Independent Set} problem restricted to $H$-free graphs admits a polynomial \lop-kernel.
\end{conjecture}

\autoref{cor:lop-kernels}, combined with the above discussion and the inapproximability result of Bonnet et al.~\cite{BonnetTTW20} on triangle-free graphs, imply the following.

\begin{corollary}\label{cor:lop4}
  The following claims hold:
  \begin{itemize}
  \item {\sc Maximum Independent Set} parameterized by the solution size does not admit a \lop-kernel with $\Ocal(k^{4 - \varepsilon})$ vertices on triangle-free graphs for any $\varepsilon > 0$, unless $\NP \subseteq \BPP$.
  \item For every real number $\varepsilon > 0$ and every integer $t \geq 3$, a \lop-kernel with $\Ocal(k^{t -1- \varepsilon})$ vertices for {\sc Maximum Independent Set} on $K_t$-free graphs would improve the best known approximation ratio $n^{\frac{t-2}{t-1}}$ that follows from Ramsey's theorem~\cite{Ramsey}.
  \item \autoref{conj:lop} implies the Improved Approximation Conjecture of Bonnet et al.~\cite{BonnetTTW20}.
  \end{itemize}
\end{corollary}

\section{An attempt to obtain a linear kernel for \textsc{MMVC}}
 \label{ap:Fernau}

 In this section we briefly explain the flaw in the linear kernel for \textsc{Maximum Minimal Vertex Cover} (\textsc{MMVC}) claimed by Fernau~\cite[Corollary 4.25]{FernauHDR}, and that is based on joint unpublished work with Dehne,
Fellows, Prieto, and Rosamond. The kernelization algorithm is a small modification of a linear kernel for the \textsc{Nonblocker Set} problem presented by Ore~\cite{Ore96}. A set of vertices $S$ of a graph $G$ is a \emph{nonblocker} if its complement is a dominating set of $G$, that is, for every $u \in S$ there exists $v \notin S$ with $\{u,v\} \in E(G)$. In the \textsc{Nonblocker Set} problem, we are given a graph $G$ and an integer parameter $k$, and the goal is to decide whether $G$ contains a nonblocker of size at least $k$. Suppose for simplicity that $G$ is connected. The idea is to consider an arbitrary spanning tree $T$ of $G$, root it arbitrarily at a vertex $r$, and partition $V(G)=V_0 \uplus V_1$ such that the vertices in $V_0$ (resp. $V_1$) are within even (resp. odd) distance from $r$ in $T$. By construction, each of $V_0$ and $V_1$ is a nonblocker in $G$, so if one of them has size at least $k$, we can answer ``\yes'', and otherwise $|V(G)| \leq 2k$ and we are done.

Back to {\sc MMVC}, it is observed in~\cite[Reduction rule 24]{FernauHDR} that a simple reduction rule allows to assume that no connected component of $G$ is a clique (in particular, an isolated vertex). Assume again for simplicity that $G$ is connected. It is then claimed in~\cite{FernauHDR} that, using the same algorithm as for  \textsc{Nonblocker Set}, the largest of $V_0$ and $V_1$, say $V_0$, can be always completed into a minimal vertex cover of $G$, which would immediately yield a kernel of size at most $2k$ for \textsc{MMVC}. Unfortunately, this claim is not true: when adding new vertices to $V_0$ in order to make it a vertex cover of $G$, we may lose the minimality property, and some vertices may need to be removed. For instance, let $G$ be the graph obtained from a triangle on vertices $u,v,w$ by adding $p \geq 2$ pendant vertices to each of $u,v$, and $w$. Let $T$ be the spanning tree obtained from $G$ by removing the edge $\{v,w\}$, and root $T$ at vertex $u$. Then $|V_0| = 1 + 2p$ and $|V_1| = 2 + p$, so $|V_0| > |V_1|$, and note that the edge $\{v,w\}$ is the only edge of $G$ not covered by $V_0$. But adding either of $v$ or $w$ to $V_0$, say $v$, results in a non-minimal vertex cover of $G$, and therefore the $p$ pendant vertices adjacent to $v$ have to be removed from $V_0$, which yields a set of size $2 + p < \frac{|V(G)|}{2} = \frac{3+3p}{2}$, where we have used that $p \geq 2$. In fact, deciding whether a set $S \subseteq V(G)$ can the extended to a minimal vertex cover of $G$ is an \NP-complete problem~\cite{CaselFGMS19}.

\section{Subquadratic kernels for \textsc{MMVC} on particular graph classes}
\label{sec:subquadratic-kernels-MMVC}

In this section we present subquadratic kernels for \textsc{Maximum Minimal Vertex Cover} restricted to particular graph classes when the parameter is the solution size $k$. Namely, in \autoref{sec:kernels-EH} we provide kernels using the Erd\H{o}s-Hajnal property, and in \autoref{sec:MMVC-other-classes} we provide further observations about other graph classes.

\subsection{Kernels using the Erd\H{o}s-Hajnal property}
\label{sec:kernels-EH}


For a constant $\delta>0$, a graph $H$ is said to satisfy the \emph{Erd\H{o}s-Hajnal property with constant~$\delta$} if every $H$-free graph $G$ on $n$ vertices contains either a clique or an independent set of size $n^{\delta}$. The (still open) Erd\H{o}s-Hajnal conjecture~\cite{ErdosH89} states that every graph $H$ satisfies the Erd\H{o}s-Hajnal property. As reported by Chudnovsky~\cite{Chudnovsky14}, the Erd\H{o}s-Hajnal conjecture has been verified for only a small number of graphs,
namely all graphs on at most four vertices, the \emph{bull} (i.e., the graph obtained by adding a pendant vertex to two different vertices of a triangle), the complete graphs, and every graph that can be constructed from them using the so-called \emph{substitution operation}~\cite{AlonPS01}, which we define later.

Since our goal is to use the Erd\H{o}s-Hajnal property in order to obtain kernels for \textsc{Maximum Minimal Vertex Cover}, we need an algorithmic version of it.  As defined by Bonnet et al.~\cite{BonnetTTW20}, for a constant $\delta>0$, a graph $H$ is said to satisfy the  \emph{constructive  Erd\H{o}s-Hajnal property with constant~$\delta$} if there exists an algorithm that takes as input an $H$-free graph $G$ on $n$ vertices, and outputs in polynomial-time a clique or an independent set of $G$ of size at least $n^{\delta}$. Fortunately for our purposes, all the graphs $H$ shown to satisfy
the Erd\H{o}s-Hajnal property so far, also satisfy its constructive version~\cite{BonnetTTW20}.

In the following simple lemma we show that, if $H$ is a graph satisfying the constructive Erd\H{o}s-Hajnal property, then the vertex set of an $H$-free graph can be partitioned in polynomial time into ``few'' cliques or independent sets. This partition will then be used to obtain subquadratic kernels on $H$-free graphs for several graphs $H$.

\begin{lemma}\label{lem:EH-partition}
Let $H$ be a graph satisfying the constructive Erd\H{o}s-Hajnal property with constant~$\delta$. The vertex set of any $H$-free graph $G$ on $n$ vertices can be partitioned in polynomial time into a collection of cliques ${\cal C}$ and a collection  of independent sets  ${\cal I}$ such that
 $|{\cal C}|+|{\cal I}| \leq \left(\frac{1}{2^{(1-\delta)}-1} \right) \cdot n^{1-\delta}$.
\end{lemma}
\begin{proof}
Let $G$ be an $H$-free graph on $n$ vertices. We initialize $X_0 = V(G), {\cal C}={\cal I}= \emptyset$, and we run the following procedure as far as $|X_0| \geq 1$:

\smallskip

\begin{adjustwidth}{1cm}{0.75cm}
 Find in polynomial time a clique or an independent set $Y$ in $G[X_0]$ with $|Y| \geq |X_0|^{\delta}$. Note that this is possible since $G[X_0]$ is an $H$-free graph for any $X_0 \subseteq V(G)$.  Add $Y$ to ${\cal C}$ or to ${\cal I}$ depending on whether $Y$ is a clique or an independent set, respectively (if $|Y|=1$, choose ${\cal C}$ or ${\cal I}$ arbitrarily).
      Update $X_0 \leftarrow X_0 \setminus Y$.
\end{adjustwidth}
\smallskip

\noindent Clearly,  the above algorithm terminates in  polynomial time.
It remains to bound $|{\cal C}|+|{\cal I}|$, which is equal to the number of iterations of the algorithm.
To this end, for a positive integer $i$, we say that an iteration belongs to \emph{step~$i$} of the algorithm if the current set $X_0$ at the start of the iteration satisfies $ \frac{n}{2^{i}} <  |X_0 | \leq \frac{n}{2^{i-1}}$.
We denote by $t_i$ the number of iterations of the algorithm within step~$i$.
By definition, $|{\cal C}|+|{\cal I}| = \sum_{i=1}^{\infty}t_i$.
Let $Y$ be a clique or an independent set found by the algorithm within step~$i$.
Since the current set $X_0$ satisfies $|X_0 | > \frac{n}{2^{i}}$, we have that $|Y| > \left(\frac{n}{2^{i}}\right)^{\delta}$. And since the sum of the sizes of the sets found before the last iteration of step $i$ is at most $\frac{n}{2^{i}}$, it follows that $t_i \leq \left(\frac{n}{2^{i}}\right)^{1-\delta}$.
Note that, in particular, $t_i = 0$ for $i > \lceil\log n \rceil$. Therefore, we conclude that
\begin{eqnarray*}
|{\cal C}|+|{\cal I}|   &=&  \sum_{i=1}^{\infty}t_i  \ \leq \ \sum_{i=1}^{\infty}\left(\frac{n}{2^{i}}\right)^{1-\delta} \ = \  n^{1-\delta} \cdot \sum_{i=1}^{\infty}   \left(\frac{1}{2^{1-\delta}}\right)^{i}
 \ = \  n^{1-\delta} \cdot \left(\frac{1}{2^{(1-\delta)}-1} \right),
\end{eqnarray*}
 and the lemma follows.
\end{proof}

We are now ready to present the subquadratic kernel on bull-free graphs. Note that, since bipartite graphs are bull-free,  \textsc{MMVC}  restricted to bull-free graphs is \NP-hard by~\cite{BoliacL03} (or by \autoref{thm:no-poly-kernel}). In the kernels presented in this section, since we can easily obtain explicit constants, we decided not to use the big-O notation.

\begin{theorem}\label{thm:kernel-bull-free}
The {\sc Maximum Minimal Vertex Cover} problem parameterized by $k$ restricted to bull-free graphs admits a kernel with at most $c  (k-1)^{7/4} + k-1$ vertices, where $c= \frac{2}{2^{\frac{3}{4}}-1} < 3$.
\end{theorem}
\begin{proof}
Let $(G,k)$ be an instance of the {\sc Maximum Minimal Vertex Cover} problem, where $G$ is a bull-free graph. Recall that by \autoref{lem:extension-nbh-is} we can assume that the maximum degree of $G$ is at most $k-1$. We start by finding greedily, starting from $V(G)$, a minimal vertex cover $X$ of $G$. Note that $X$ can be easily found in polynomial time by \autoref{obs:characterization-mvc}. If $|X| \geq k$, we conclude that  $(G,k)$ is a \yes-instance, so we can assume  that $|X| \leq k-1$.  Let $S= V(G) \setminus X$ and note that $S$ is an independent set.

Since the bull satisfies the constructive Erd\H{o}s-Hajnal property with constant~$\delta=\frac{1}{4}$~\cite{ChudnovskyS08d,BonnetTTW20}, we can apply \autoref{lem:EH-partition} to the bull-free graph $G[X]$ and obtain in polynomial time a partition of $X$ into a collection of cliques ${\cal C}$ and a collection  of independent sets  ${\cal I}$ such that $|{\cal C}|+|{\cal I}| \leq d \cdot |X|^{3/4} \leq d \cdot  (k-1)^{3/4}$, where $d= \frac{1}{2^{\frac{3}{4}}-1} < 1.47$. Since we can assume that $G$ has no isolated vertices, as they can be safely removed without affecting the type of the instance, it follows that
\begin{equation}\label{eq:bound-bull-1}
S = \bigcup_{C \in {\cal C}}N_{S}(C) \cup \bigcup_{I \in {\cal I}}N_{S}(I).
\end{equation}
Hence,  our objective is to bound $|N_{S}(Y)|$ for every $Y\in {\cal C} \cup {\cal I}$. Suppose first that $I\in  {\cal I}$ is an independent set. From \autoref{lem:extension-nbh-is}, if $|N_S(I)| \geq k$ we can conclude that $(G,k)$ is  a \yes-instance, so we can assume henceforth that
\begin{equation}\label{eq:bound-bull-2}
\text{for every independent set  $I\in {\cal I}$,  it holds }\ |N_S(I)| \leq k-1.
\end{equation}
Suppose now that $C\in  {\cal C}$ is a clique. We partition $N_S(C)=S_C^{1} \uplus S_C^{2}$ as follows. Let $S_C^{1}$ be an inclusion-wise maximal set of vertices in $N_S(C)$ such that for any two (not necessarily distinct) vertices $x,y \in S_C^{1}$, $|N_C(x) \cup N_C(y)| \leq |C|-1$. That is, $S_C^{1}$ is a maximal set in $N_S(C)$ such that the neighborhoods of its vertices pairwise do {\sl not} cover the whole clique $C$. We let  $S_C^{2} = N_S(C) \setminus S_C^{1}$. The following is the crucial property of the set $S_C^{1}$.
\begin{claim}\label{claim:bull-free-inclusion}
The vertices in $S_C^{1}$ can be ordered $x_1,\ldots,x_p$ so  that $N_C(x_i) \subseteq N_C(x_j)$ whenever $i \leq j$.
\end{claim}
\begin{proof}
In order to prove the claim, it is sufficient to prove that, for any two vertices $x,y \in S_C^{1}$, either $N_C(x) \subseteq N_C(y)$ or $N_C(y) \subseteq N_C(x)$. Suppose for the sake of contradiction that there exist two vertices $u \in N_C(x) \setminus N_C(y)$ and $v \in N_C(y) \setminus N_C(x)$. By definition of the set $S_C^{1}$, there exists a vertex $w \in C \setminus (N_C(x) \cup N_C(y))$. But then the vertices $x,y,u,v,w$ induce a bull as illustrated in \autoref{fig:bull}, contradicting the hypothesis that $G$ is bull-free.
\end{proof}
\begin{figure}[htb]
\begin{center}
\includegraphics[scale=.9]{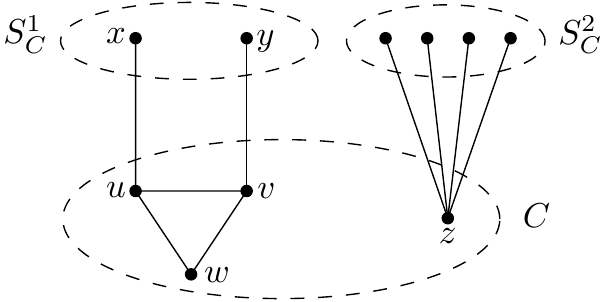}
\end{center}
\caption{Configuration considered in the proof of \autoref{claim:bull-free-inclusion} and a vertex $z \in \bigcap_{x \in S_C^{2}}N_C(x)$.}
\label{fig:bull}
\end{figure}
\autoref{claim:bull-free-inclusion} implies in particular that, unless $S_C^{1}=\emptyset$,  there exists a vertex $u \in \bigcap_{x \in S_C^{1}}N_C(x)$. Since $u$ has degree at most $k-1$ in $G$, and each vertex $x \in S_C^{1}$ is adjacent to $u$, it follows that $|S_C^{1}| \leq k-1$.

Let us now focus on the set $S_C^{2}$. The definition of the set $S_C^{1}$ together with \autoref{claim:bull-free-inclusion}  imply that
there exists a vertex
$z \in C \setminus \bigcup_{y \in S_C^{1}}N_C(y)$. Consider now an arbitrary vertex $x \in S_C^{2}$. Since $x$ could not be added to $S_C^{1}$, there exists a vertex $y \in S_C^{1}$ such that $N_C(x) \cup N_C(y) = C$. But since $z \in C \setminus \bigcup_{y \in S_C^{1}}N_C(y)$, necessarily $z \in N_C(x)$. It follows that $z \in \bigcap_{x \in S_C^{2}}N_C(x)$ (see \autoref{fig:bull}). Using again the fact that $z$ has degree at most $k-1$ in $G$, we obtain that $|S_C^{2}| \leq k-1$. Summarizing, we have that
\begin{equation}\label{eq:bound-bull-3}
\text{for every clique  $C\in {\cal C}$,  it holds }\ |N_S(C)| = |S_C^{1}| + |S_C^{2}| \leq 2(k-1).
\end{equation}
Putting all pieces together, Equations~(\ref{eq:bound-bull-1}),~(\ref{eq:bound-bull-2}), and~(\ref{eq:bound-bull-3}) and the fact that $|X| \leq k-1$ and $|{\cal C}|+|{\cal I}| \leq d \cdot |X|^{3/4}$ imply that, unless we have already concluded that $(G,k)$ is a \yes-instance, we have that
\begin{eqnarray*}
|V(G)| & = &  |X| + |S| \ = \ |X| + |\bigcup_{C \in {\cal C}}N_{S}(C)| + |\bigcup_{I \in {\cal I}}N_{S}(I)|\\
& \leq & |X| + (|{\cal C}|+|{\cal I}|) \cdot \max_{Y \in {\cal C}\cup{\cal I}}|N_S(Y)| \ \leq \ k-1 + d \cdot  (k-1)^{3/4} \cdot 2(k-1)\\
& = & 2d \cdot (k-1)^{7/4} + k-1,
\end{eqnarray*}
and the theorem follows.
\end{proof}

It is easy to prove that, for every integer $t\geq 2$, every $K_t$-free graph $G$ on $n$ vertices has an independent set of size $n^{\frac{1}{t-1}}$, by induction on $t$: for $t=2$ the statement is trivial, and if $t \geq 3$, then either $\Delta(G) < n^{\frac{t-2}{t-1}}$, and an independent set of size $n^{\frac{1}{t-1}}$ can be found greedily by adding any vertex to it and deleting its neighborhood, or there exists a vertex $v \in V(G)$ of degree at least $n^{\frac{t-2}{t-1}}$, in which case an independent set of size $n^{\frac{1}{t-1}}$ can be found applying the inductive hypothesis to the $K_{t-1}$-free graph $G[N(v)]$. Clearly, this proof can be translated to a polynomial-time algorithm to find an independent set of the appropriate size. Therefore, for any integer $t \geq 2$, $K_t$ satisfies the constructive Erd\H{o}s-Hajnal property with constant~$\delta = \frac{1}{t-1}$. The proof of the following theorem is a simplified version of that of \autoref{thm:kernel-bull-free}. Note that, since bipartite graphs are $K_t$-free for every $t \geq 3$, {\sc MMVC} is \NP-hard on $K_t$-free graphs~\cite{BoliacL03}.

\begin{theorem}\label{thm:kernel-Kt-free}
For every integer $t\geq 3$, the {\sc Maximum Minimal Vertex Cover} problem parameterized by $k$ restricted to $K_t$-free graphs admits a kernel with at most $c_{t} (k-1)^{\frac{2t-3}{t-1}} + k-1$ vertices, where $c_t =\frac{t-1}{2^{\frac{t-2}{t-1}} -1}$.
\end{theorem}
\begin{proof}
As in the proof of \autoref{thm:kernel-bull-free}, given an instance $(G,k)$ of
{\sc Maximum Minimal Vertex Cover}, where $G$ is a $K_t$-free graph, we partition $V(G) = X \uplus S$, where $X$ is a minimal vertex cover of $G$ with $|X| \leq k-1$, and we use \autoref{lem:EH-partition} to partition $X$ into two collections $\mathcal{C}$ and $\mathcal{I}$ of cliques and independent sets, respectively, with $|{\cal C}|+|{\cal I}| \leq d_{t} \cdot |X|^{\frac{t-2}{t-1}}$, where $d_{t} = \frac{1}{2^{\frac{t-2}{t-1}} -1}$. Equations~(\ref{eq:bound-bull-1}) and~(\ref{eq:bound-bull-2}) still hold, but now we have a much simpler version of Equation~(\ref{eq:bound-bull-3}): if $C \in {\cal C}$ is a clique then, since $G$ is $K_t$-free, necessarily $|C| \leq t-1$, which together with the fact that $\Delta(G) \leq k-1$ yield
\begin{equation}\label{eq:Kt-free-3}
\text{for every clique  $C\in {\cal C}$,  it holds }\ |N_S(C)| =  (t-1)(k-1).
\end{equation}
Combining Equations~(\ref{eq:bound-bull-1}),~(\ref{eq:bound-bull-2}), and~(\ref{eq:Kt-free-3}) we get
\begin{eqnarray*}
|V(G)| & \leq &  |X| + (|{\cal C}|+|{\cal I}|) \cdot \max_{Y \in {\cal C}\cup{\cal I}}|N_S(Y)| \ \leq \ k-1 + d_t \cdot (k-1)^{\frac{t-2}{t-1}} \cdot (t-1)(k-1),
\end{eqnarray*}
and the theorem follows.
\end{proof}

We now extend the results of \autoref{thm:kernel-bull-free} and \autoref{thm:kernel-Kt-free} to more general excluded induced graphs $H$, by making use of the aforementioned substitution operation.
As defined by Alon et al.~\cite{AlonPS01}, for two graphs $H_1$ and $H_2$ on disjoint vertex sets, we say that $H$ is \emph{obtained from $H_1$ by substituting $H_2$ for $v \in V(H_1)$} (or just \emph{obtained from $H_1$ by substituting $H_2$} if the vertex $v$ in question is not important) if
\begin{itemize}
  \item $V(H)=(V(H) \setminus \{v\}) \cup V(H_2)$,
  \item $H[V(H_2)]=H_2$,
  \item $H[V(H_1) \setminus \{v\}] = H_1 \setminus v$, and
  \item $u \in V(H_1)$ is adjacent in $H$ to $w \in V(H_2)$ if and only if $u$ is adjacent in $H_1$ to $v$.
\end{itemize}

Alon et al.~\cite{AlonPS01} proved that if two graphs $H_1$ and $H_2$ satisfy Erd\H{o}s-Hajnal property and $H$ is obtained from $H_1$ by substituting $H_2$, then  $H$ satisfies the Erd\H{o}s-Hajnal property as well. More precisely, by following the details in the proof of~\cite[Theorem 2.1]{AlonPS01}, we can derive that if  $H_1$ and $H_2$ satisfy Erd\H{o}s-Hajnal property with constants $\delta_1$ and $\delta_2$, respectively, then $H$ satisfies the Erd\H{o}s-Hajnal property with constant $\delta = \frac{\delta_2}{\delta_1 + |V(H_1)| \cdot \delta_2}$. The same applies to the constructive version of the Erd\H{o}s-Hajnal property.

For an integer $t \geq 2$, we define the \emph{$t$-bull} as the graph obtained from $K_t$ by adding a pendant vertex to two different vertices of the clique. Note that the 2-bull is equal to $P_4$ and that the 3-bull is equal to the bull. 
Note also that, for every $t \geq 3$, the $t$-bull is obtained from the bull by substituting $K_{t-2}$ for the vertex of degree two of the bull. Therefore, by the discussion in the above paragraph, since the bull and $K_{t-2}$ satisfy the constructive Erd\H{o}s-Hajnal property with constants $\frac{1}{4}$ and $\frac{1}{t-3}$, respectively,  it follows that, for every $t \geq 4$, the $t$-bull satisfies the constructive Erd\H{o}s-Hajnal property with constant
$$
\delta_t\ =\ \frac{\frac{1}{t-3}}{\frac{1}{4} + \frac{5}{t-3}} \ = \ \frac{4}{t+17}.
$$

The proof of the next theorem follows again (and generalizes) that of \autoref{thm:kernel-bull-free}.
Note that \autoref{thm:kernel-fat-bull-free} corresponds to the particular case $t=3$ of \autoref{thm:kernel-Kt-free}. Note also that bipartite graphs are $t$-bull-free for $t \geq 3$, hence {\sc MMVC} is \NP-hard on $t$-bull-free graphs for $t \geq 3$~\cite{BoliacL03}.
On the other hand, $2$-bull-free graphs are exactly $P_4$-free graphs, also called cographs, which have cliquewidth at most two, hence by \autoref{obs:kernel_MSO} (proved later in \autoref{sec:MMVC-other-classes}) {\sc MMVC} can be solved in polynomial time on this class.

\begin{theorem}\label{thm:kernel-fat-bull-free}
For every integer $t\geq 3$, the {\sc Maximum Minimal Vertex Cover} problem parameterized by $k$ restricted to $t$-bull-free graphs admits a kernel with at most $c_t (k-1)^{2-\delta_t} + k-1$ vertices, where  $\delta_3= \frac{1}{4}$ and $\delta_t = \frac{4}{t+17}$ for $t \geq 4$, and $c_{t} = \frac{t-1}{2^{(1-\delta_t)}-1}$.
\end{theorem}
\begin{proof}
As in the proof of \autoref{thm:kernel-bull-free}, given an instance $(G,k)$ of
{\sc Maximum Minimal Vertex Cover}, where $G$ is a $t$-bull-free graph, we partition $V(G) = X \uplus S$, where $X$ is a minimal vertex cover of $G$ with $|X| \leq k-1$, and we use \autoref{lem:EH-partition} to partition $X$ into two collections $\mathcal{C}$ and $\mathcal{I}$ of cliques and independent sets, respectively, with $|{\cal C}|+|{\cal I}| \leq d_{t} \cdot |X|^{1-\delta_t}$, where $\delta_3= \frac{1}{4}$ and $\delta_t = \frac{4}{t+17}$ for $t \geq 4$, and $d_{t} = \frac{1}{2^{(1-\delta_t)}-1}$ for every $t\geq 3$. Equations~(\ref{eq:bound-bull-1}) and~(\ref{eq:bound-bull-2}) still hold for every integer $t \geq 3$, but now we need slightly more involved arguments to obtain an appropriate version of Equation~(\ref{eq:bound-bull-3}) for every $t \geq 3$.

To this end, suppose that $C\in  {\cal C}$ is a clique. We partition $N_S(C)=S_C^{1} \uplus S_C^{2}$ as follows. Let $S_C^{1}$ be an inclusion-wise maximal set of vertices in $N_S(C)$ such that for any two (not necessarily distinct) vertices $x,y \in S_C^{1}$, $|N_C(x) \cup N_C(y)| \leq |C|-(t-2)$. That is, $S_C^{1}$ is a maximal set in $N_S(C)$ such that the neighborhoods of its vertices pairwise leave at least $t-2$ uncovered vertices in the clique $C$. We let  $S_C^{2} = N_S(C) \setminus S_C^{1}$. The set $S_C^{1}$ satisfies exactly the same crucial property as for the case $t=3$ (see \autoref{claim:bull-free-inclusion}).
\begin{claim}\label{claim:t-bull-free-inclusion}
For every integer $t \geq 3$, the vertices in $S_C^{1}$ can be ordered $x_1,\ldots,x_p$ so  that $N_C(x_i) \subseteq N_C(x_j)$ whenever $i \leq j$.
\end{claim}
\begin{proof}
In order to prove the claim, it is sufficient to prove that, for any two vertices $x,y \in S_C^{1}$, either $N_C(x) \subseteq N_C(y)$ or $N_C(y) \subseteq N_C(x)$. Suppose for the sake of contradiction that there exist two vertices $u \in N_C(x) \setminus N_C(y)$ and $w \in N_C(y) \setminus N_C(x)$. By definition of the set $S_C^{1}$, there exist $t-2$ vertices $w_1,\ldots,w_{t-2} \in C \setminus (N_C(x) \cup N_C(y))$. But then the vertices $x,y,u,v,w_1,\ldots,w_{t-2}$ induce a $t$-bull, contradicting the hypothesis that $G$ is $t$-bull-free.
\end{proof}
\autoref{claim:t-bull-free-inclusion} implies in particular that, unless $S_C^{1}=\emptyset$,  there exists a vertex $u \in \bigcap_{x \in S_C^{1}}N_C(x)$. Since $u$ has degree at most $k-1$ in $G$, and each vertex $x \in S_C^{1}$ is adjacent to $u$, it follows that $|S_C^{1}| \leq k-1$.

Let us now focus on the set $S_C^{2}$. The definition of the set $S_C^{1}$ together with \autoref{claim:t-bull-free-inclusion}  imply that
there exist at least $t-2$
vertices $z_1,\ldots,z_{t-2} \in C \setminus \bigcup_{y \in S_C^{1}}N_C(y)$.  Consider now an arbitrary vertex $x \in S_C^{2}$. Since $x$ could not be added to $S_C^{1}$, there exists a vertex $y \in S_C^{1}$ such that $|N_C(x) \cup N_C(y)| \geq |C|-(t-1)$. But since $z_1,\ldots,z_{t-2} \in C \setminus \bigcup_{y \in S_C^{1}}N_C(y)$, there exists an index $j \in [t-2]$ such that $z_j \in N_C(x)$. That is, every vertex $x \in S_C^{2}$ is adjacent to at least one of the vertices
$z_1,\ldots,z_{t-2}$. Using again the fact that each of the vertices $z_1,\ldots,z_{t-2}$ has degree at most $k-1$ in $G$, we obtain that $|S_C^{2}| \leq (t-2)(k-1)$. Summarizing, we have that
\begin{equation}\label{eq:bound-t-bull-3}
\text{for every clique  $C\in {\cal C}$,  it holds }\ |N_S(C)| = |S_C^{1}| + |S_C^{2}| \leq (t-1)(k-1).
\end{equation}
Putting all pieces together, Equations~(\ref{eq:bound-bull-1}),~(\ref{eq:bound-bull-2}), and~(\ref{eq:bound-t-bull-3}) and the fact that $|X| \leq k-1$ and $|{\cal C}|+|{\cal I}| \leq d_{t} \cdot |X|^{1-\delta_t}$ imply that, unless we have already concluded that $(G,k)$ is a \yes-instance, we have that
\begin{eqnarray*}
|V(G)| & \leq &  |X| + (|{\cal C}|+|{\cal I}|) \cdot \max_{Y \in {\cal C}\cup{\cal I}}|N_S(Y)| \ \leq \ k-1 + d_{t} \cdot (k-1)^{1-\delta_t} \cdot (t-1)(k-1),
\end{eqnarray*}
and the theorem follows.
\end{proof}

Let the \emph{paw} be the graph obtained from a triangle by adding a pendant edge.
Gy{\'{a}}rf{\'{a}}s~\cite{Gyarfas13} showed that the paw satisfies the constructive Erd\H{o}s-Hajnal property with constant $\delta = \frac{1}{3}$.
Note that bipartite graphs are paw-free, hence {\sc MMVC} is \NP-hard on paw-free graphs~\cite{BoliacL03}.


\begin{theorem}\label{thm:kernel_paw}
The {\sc Maximum Minimal Vertex Cover} problem parameterized by $k$ restricted to paw-free graphs admits a kernel with at most $c  (k-1)^{5/3} + k-1$ vertices, where $c=\frac{2}{2^{2/3}-1} < 3.41$.
\end{theorem}
\begin{proof}
Given an instance $(G,k)$ of
{\sc Maximum Minimal Vertex Cover}, where $G$ is a paw-free graph, we again partition $V(G) = X \uplus S$, where $X$ is a minimal vertex cover of $G$ with $|X| \leq k-1$,
and we use \autoref{lem:EH-partition} to partition $X$ into two collections $\mathcal{C}$ and $\mathcal{I}$ of cliques and independent sets, respectively, with $|{\cal C}|+|{\cal I}| \leq d \cdot |X|^{2/3}$, where $d=\frac{1}{2^{2/3}-1}$.
Equations~(\ref{eq:bound-bull-1}) and~(\ref{eq:bound-bull-2}) still hold, and we can again obtain in a simpler way an appropriate version of Equation~(\ref{eq:bound-bull-3}). Indeed, let  $C \in {\cal C}$ be a clique, and our goal is to bound $|N_S(C)|$. If $|C|=1$ then by the fact that $\Delta(G) \leq k-1$ we get that $|N_S(C)| \leq k-1$, so assume that $|C|\geq 2$. Suppose for the sake of contradiction that there exists a vertex $v \in N_S(C)$ such that $|N_C(v)| \leq |C|-2$. Let $w \in N_C(v)$ and let $z_1,z_2$ be two vertices in $C \setminus N_C(v)$. Then the vertices $v,w,z_1,z_2$ induce a paw, contradicting the hypothesis that $G$ is paw-free. Therefore,  for every vertex $v \in N_S(C)$ it holds that $|N_C(v)| \geq |C|-1$. Hence, the number of edges in $G$ between $C$ and $N_S(C)$ is at least $|N_S(C)| \cdot (|C|-1)$ and, since $\Delta(G) \leq k-1$, at most $|C| \cdot (k-1)$. Using that $|C| \geq 2$, it follows that
\begin{equation}\label{eq:bound-paw-free-3}
\text{for every clique  $C\in {\cal C}$,  it holds }\ |N_S(C)| \leq \frac{|C|}{|C|-1} \cdot (k-1)  \leq  2(k-1).
\end{equation}
Putting all pieces together, Equations~(\ref{eq:bound-bull-1}),~(\ref{eq:bound-bull-2}), and~(\ref{eq:bound-paw-free-3}) and the fact that $|X| \leq k-1$ and $|{\cal C}|+|{\cal I}| \leq d \cdot |X|^{2/3}$ imply that, unless we have already concluded that $(G,k)$ is a \yes-instance, we have that
\begin{eqnarray*}
|V(G)| & \leq &  |X| + (|{\cal C}|+|{\cal I}|) \cdot \max_{Y \in {\cal C}\cup{\cal I}}|N_S(Y)| \ \leq \ k-1 + d \cdot (k-1)^{2/3} \cdot 2(k-1),
\end{eqnarray*}
and the theorem follows.
%
%
\end{proof}

\subsection{Remarks on other graph classes}
\label{sec:MMVC-other-classes}

In this subsection we provide additional observations about the complexity of
the {\sc Maximum Minimal Vertex Cover} problem restricted to special graph classes.

\begin{lemma}\label{lem:kernel_K1t}
For every integer $t\geq 1$, the {\sc Maximum Minimal Vertex Cover} problem parameterized by $k$ restricted to $K_{1,t}$-free graphs admits a kernel with at most $t(k-1)$ vertices.
\end{lemma}
\begin{proof}
Given an instance $(G,k)$ of
{\sc Maximum Minimal Vertex Cover}, where $G$ is a $K_{1,t}$-free graph, we again partition $V(G) = X \uplus S$, where $X$ is a minimal vertex cover of $G$ with $|X| \leq k-1$. Since $G$ is $K_{1,t}$-free and $S$ is an independent set, it holds that for every $v \in X$, $|N_S(v)| \leq t-1$, and since we can assume that $G$ contains no isolated vertex, we obtain that $|V(G)| = |X| + |\bigcup_{v \in X}|N_S(v)| \leq k-1 + (t-1)(k-1) = t(k-1)$.
\end{proof}

Let ${\cal C}$ be a graph class such that there exists a polynomial-time algorithm that, given a graph $G \in {\cal C}$, outputs a proper coloring of the vertices of $G$ using at most $c$ colors, for some integer $c \geq 1$. We say that such a graph class ${\cal C}$ is \emph{{\sf poly}-$\chi$-$c$-bounded}.
Examples of {\sf poly}-$\chi$-$c$-bounded classes are planar graphs, minor-free graphs, or, more generally, graphs of bounded expansion. We note that Fernau~\cite[Corollary 4.14]{FernauHDR}  provides a similar observation for the particular case of planar graphs.

\begin{lemma}\label{lem:polyChi}
For every integer $c \geq 1$, the {\sc Maximum Minimal Vertex Cover} problem parameterized by $k$ restricted to the class of {\sf poly}-$\chi$-$c$-bounded graphs admits a kernel with at most $c(k-1)$ vertices.
\end{lemma}
\begin{proof}
Given an instance $(G,k)$ of {\sc MMVC}, where $G$ belongs to a {\sf poly}-$\chi$-$c$-bounded class, we first compute in polynomial time a proper vertex coloring of $G$ using at most $c$ colors. We may clearly assume that $G$ has no isolated vertices, as such vertices can be safely removed. Let $V(G) = S_1 \uplus \cdots \uplus S_c$ be the corresponding partition of $V(G)$ into independent sets. By \autoref{lem:extension-nbh-is}, for every $i \in [c]$ there exists a minimal vertex cover of $G$ that contains $N(S_i)$. Hence, if for some $i \in [c]$ we have that $|N(S_i)| \geq k$, we can safely answer ``\yes'', so we may assume that, for every $i \in [c]$, $|N(S_i)| \leq k-1$. Since $G$ has no isolated vertices and every set $S_i$ is an independent set, it follows that $V(G) = \bigcup_{i \in [c]} N(S_i)$, so we have that $|V(G)| \leq \sum_{i \in [c]}|N(S_i)|\leq c(k-1)$.
\end{proof}

Another graph class ${\cal K}$ that allows for linear kernels is defined such that, for every graph $G \in {\cal K}$, the minimum size of a dominating set of $G$ is equal to the size of a minimum {\sl independent} dominating set of $G$. We furthermore ask ${\cal K}$ to be hereditary. Such graphs have been studied, for instance, in~\cite{ToppV91,AllanL78}, and include in particular $K_{1,3}$-free graphs (note that a generalization to $K_{1,t}$-graphs is given in \autoref{lem:kernel_K1t}). Let us see why the class ${\cal K}$ allows for a linear kernel. As discussed at the end of \autoref{sec:nopolykernel}, the complement of a dominating set is called a nonblocker, and the \textsc{Nonblocker Set} problem admits a linear kernel~\cite{FernauHDR}. On the other hand, the complement of an independent dominating set is a minimal vertex cover. Hence, if $G \in {\cal K}$, an instance $(G,k)$ of \textsc{Nonblocker Set} is positive if and only if $(G,k)$ is a positive instance of {\sc MMVC}. Note the linear kernel for the \textsc{Nonblocker Set} problem discussed at the end of \autoref{sec:nopolykernel} outputs a subgraph $G'$ of $G$, and we have that $G' \in {\cal K}$ since ${\cal K}$ is hereditary. Hence, the equivalence between \textsc{Nonblocker Set} and {\sc MMVC} also holds for $G'$, and it follows that this kernel is also a linear kernel for {\sc MMVC} restricted to graphs in ${\cal K}$.

%
%

Our last contribution in this section concerns graph classes of bounded \emph{cliquewidth}. Cliquewidth, which we do not need to define here, is a graph parameter that is ``smaller'' than treewidth in the sense that graph classes of bounded treewidth have also bounded cliquewidth (the opposite is not true, as cliques have cliquewidth one but unbounded treewidth); see~\cite{CourcelleMR00} for the formal definition.

The variation of \emph{monadic second order} logic of graphs called \MSOone is defined by a syntax that includes the logical connectives $\vee$, $\wedge$, $\neg$, variables for vertices, edges, sets of vertices (but {\sl not} sets of edges), the quantifiers $\forall, \exists$ that can be applied to these variables, and the binary relations expressing whether a vertex  belongs to a set, whether an edge is incident to vertex, whether two vertices are adjacent, and whether two sets are equal. It is well-known that
finding a minimum or maximum weight vertex set that satisfies a given graph property expressed in \MSOone can be solved in linear time on graphs of cliquewidth bounded by a constant~\cite{CourcelleMR00,ArnborgLS91}.

\begin{observation}\label{obs:kernel_MSO}
The {\sc Maximum Minimal Vertex Cover} problem can be expressed in \MSOone, and therefore it can be solved in linear time when restricted to any graph class of cliquewidth bounded by a constant.
\end{observation}
\begin{proof}
Given a graph $G$, we can express the property of a vertex set $S$ being a minimal vertex cover of $G$ in the syntax of \MSOone as follows: for every pair of vertices $u,v$ such that $u$ is adjacent to $v$, $u \in S$ or $v \in S$ (this guarantees that $S$ is a vertex cover of $G$), and for every vertex $v \in V(G)$,  $v \notin S$ or there exists a vertex $u$ adjacent to $v$ such that $u \notin S$ (this guarantees, by \autoref{obs:characterization-mvc}, that $S$ is minimal).
\end{proof}



Let the \emph{diamond} be the graph obtained from $K_4$ by removing an edge.
Since Brandst{\"{a}}dt~\cite{Brandstadt04} proved that $\{P_5,\text{diamond}\}$-free graphs have bounded cliquewidth, from \autoref{obs:kernel_MSO} we immediately get the following corollary.

\begin{corollary}\label{cor:diamond-P5}
The {\sc Maximum Minimal Vertex Cover} problem restricted to $\{P_5,\text{diamond}\}$-free graphs can be solved in linear time.
\end{corollary}




\section{Ruling out polynomial kernels for \textsc{MMVC} for smaller parameters}
\label{sec:nopolykernel}

In this section we rule out, assuming that ${\sf NP} \nsubseteq {\sf coNP} / {\sf poly}$, the existence of polynomial kernels for {\sc MMVC} parameterized by the size of a minimum vertex cover of the input graph. As mentioned in the introduction, the reduction given in \autoref{thm:no-poly-kernel} also provides an alternative proof of the \NP-completeness of {\sc MMVC} on bipartite graphs, which also follows from~\cite{BoliacL03}. We note that the existing \NP-hardness reductions for {\sc MMVC}, such as the one in~\cite{BoliacL03},  do {\sl not} seem to be easily modifiable so to yield the non-existence of polynomial kernels.



\begin{theorem}\label{thm:no-poly-kernel}
The {\sc Maximum Minimal Vertex Cover} problem parameterized by the size of a minimum vertex cover (or of a maximum matching) of the input graph does not admit a polynomial kernel unless ${\sf NP} \subseteq {\sf coNP} / {\sf poly}$, even restricted to bipartite graphs.
\end{theorem}
\begin{proof}
We present a \PPT from \textsc{Monotone Sat}  parameterized by the number of variables, which is also an \NP-completeness reduction. The \textsc{Monotone Sat} problem is the restriction of the \textsc{Sat} problem to formulas in which the literals in each clause are either all positive or all negative. This problem is well-known to be \NP-complete~\cite{GareyJ79}, and it is easy to see that, when parameterized by the number of variables, it does not admit a polynomial kernel unless ${\sf NP} \subseteq {\sf coNP} / {\sf poly}$. Indeed,  Fortnow and Santhanam~\cite{FortnowS11} proved that
the \textsc{Sat} problem parameterized by the number of variables does not admit a polynomial kernel unless ${\sf NP} \subseteq {\sf coNP} / {\sf poly}$, and the classical reduction from \textsc{Sat} to \textsc{Monotone Sat} that replaces each variable with a ``positive'' and a ``negative'' variable and adds extra clauses appropriately~\cite{GareyJ79} is in fact a \PPT when the parameter is the number of variables.

Given an instance $\phi$ of \textsc{Monotone Sat}, where the formula $\phi$ contains $n$ variables and $m$ clauses, we construct in polynomial time an instance $(G,k)$ of
{\sc Maximum Minimal Vertex Cover} as follows. For each variable $x_i$ of $\phi$, $i \in [n]$, we add to $G$ four vertices $\ell_i,x_i^+,x_i^-,r_i$ and three edges $\{\ell_i,x_i^+\}, \{x_i^+,x_i^-\},\{x_i^-,r_i\}$, hence inducing a $P_4$. We call the vertex $x_i^+$ (resp. $x_i^-$) a \emph{positive} (resp. a \emph{negative}) vertex of $G$. For each clause $C_j$ of $\phi$, $j \in [m]$, we add to $G$ a vertex $c_j$, which we connect to the positive or negative vertices corresponding to the literals contained in $C_j$. This concludes the construction of $G$, which is illustrated in \autoref{fig:reduction}(a). Note that, since $\phi$ is a monotone formula, $G$ is a bipartite graph. Note also that the set of vertices $\{x_i^+,x_i^- \mid i \in [n]\}$ is a minimum vertex cover of $G$ of size $2n$, and that the set of edges $\{\{\ell_i,x_i^+\},\{x_i^-,r_i\} \mid i \in [n]\}$ is a maximum matching of $G$ of size $2n$.  We claim that $\phi$ is satisfiable if and only if $G$ contains a minimal vertex cover of size $k:=2n+m$.

\begin{figure}[h]
\begin{center}
\vspace{-.25cm}
\includegraphics[scale=.93]{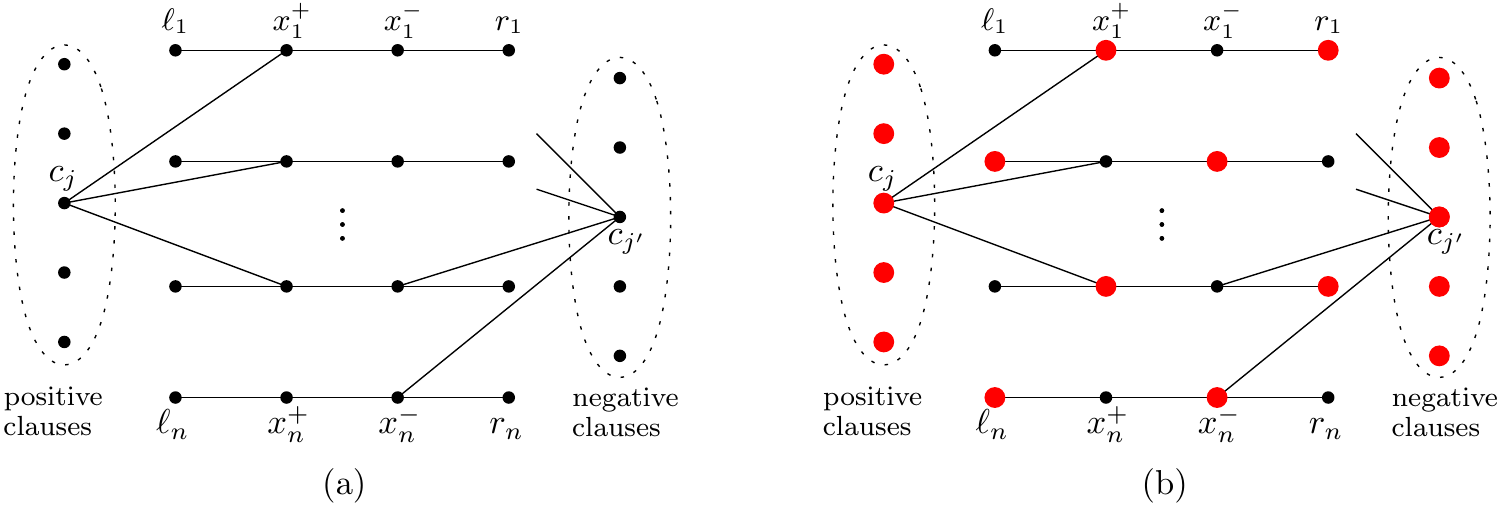}
\end{center}
\vspace{-.25cm}
\caption{(a) Illustration of the graph $G$ built from the formula $\phi$ in the proof of \autoref{thm:no-poly-kernel}. (b)~A minimal vertex cover $X$ of $G$ is shown with larger red vertices.}
\label{fig:reduction}
\vspace{-.15cm}
\end{figure}

Suppose first that $\phi$ is satisfiable, and let $\sigma$ be an assignment of the variables that satisfies all the clauses in $\phi$. We proceed to define a minimal vertex cover $X$ of $G$ of size $k$. First, add to $X$ all the clause vertices $\{c_j \mid j \in [m]\}$. For every $i \in [n]$, if $\sigma(x_i)={\sf true}$ (resp. $\sigma(x_i)={\sf false}$), add to $X$ vertices $x_i^-$ and $\ell_i$ (resp. $x_i^+$ and $r_i$). See \autoref{fig:reduction}(b) for an illustration, where the set $X$ is shown with larger red vertices. Clearly, $X$ is a vertex cover of $G$. To see that it is minimal, by \autoref{obs:characterization-mvc} it is enough to verify that, for every vertex $v \in X$, $N[v] \nsubseteq X$. This condition holds easily for all vertices in $X$ that are in the $P_4$'s, since for each $P_4$ its vertices in $X$ are not adjacent. Let $c_j$ be a clause vertex.
Since $\sigma$ is a satisfying assignment of the variables, there exists a variable $x_i$ such that if $\sigma(x_i)={\sf true}$ (resp. $\sigma(x_i)={\sf false}$) then $x_i \in C_j$ (resp. $\bar{x}_i \in C_j$). By definition of $X$, if $\sigma(x_i)={\sf true}$ (resp. $\sigma(x_i)={\sf false}$) then $x_i^+ \notin X$ (resp. $x_i^- \notin X$), and by construction of $G$ we have that $x_i^+ \in N(c_j)$ (resp. $x_i^- \in N(c_j)$), so in both cases $N[c_j] \nsubseteq X$.

Conversely, suppose that $G$ contains a minimal vertex cover $X$ of size $k$, and we proceed to define a variable assignment $\sigma$ as follows. For $i \in [n]$, as $\{x_i^+,x_i^-\} \in E(G)$ we have that $X$ contains one or two vertices in the set  $\{x_i^+,x_i^-\}$. If $x_i^+ \notin X$ (resp. $x_i^- \notin X$) we set $\sigma(x_i) = {\sf true}$ (resp. $\sigma(x_i) = {\sf false}$), and if both $x_i^+$ and $x_i^-$ belong to $X$ we set $\sigma(x_i)$ to {\sf true} or to {\sf false} arbitrarily. We claim that $\sigma$ satisfies all the clauses in $\phi$. For $i \in [n]$, let $P^i$ be the $P_4$ of $G$ induced by the vertices $\ell_i,x_i^+,x_i^-,r_i$. Since $X$ is a vertex cover, clearly $|X \cap V(P^i)| \geq 2$. We claim that $|X \cap V(P^i)| = 2$. Indeed, if $|X \cap V(P^i)| \geq 3$, then $\{\ell_i,x_i^+\}\subseteq X$ or $\{x_i^-,r_i\} \subseteq X$ (or both). But then $N[\ell_i] \subseteq X$ or $N[r_i] \subseteq X$ (or both), contradicting \autoref{obs:characterization-mvc}. Thus, $|X \cap V(P^i)| = 2$, which implies that $|X \cap \bigcup_{i \in [n]} V(P^i)| = 2n$, hence necessarily $X$ contains the whole set $\{c_j \mid j \in [m]\}$ of clause vertices. Consider an arbitrary clause vertex $c_j$. Since $X$ is minimal and $c_j \in X$, by \autoref{obs:characterization-mvc} there exists a neighbor of $c_j$ in $G$ that is {\sl not} in $X$, and by definition of $\sigma$ it follows that the literal corresponding to that neighbor of $c_j$ satisfies clause $C_j$. Thus, $\sigma$ is a satisfying assignment and the proof is complete.

Finally, note that the above reduction is also an \NP-completeness reduction from \textsc{Monotone Sat} to {\sc Maximum Minimal Vertex Cover} on bipartite graphs.
\end{proof}

\section{Conclusions and further research}
\label{sec:concl}

Motivated by the existence of subquadratic kernels for the \textsc{Maximum Minimal Vertex Cover} problem parameterized by the solution size, we introduced a general framework to rule out certain types of kernels, which we called \lop-kernels, for optimization problems.
This ``\lop'' assumption does not seem to be very restrictive, as the vast majority of known kernels are in fact \lop-kernels~\cite{Book-kernels}. For instance, the classical kernels for \textsc{Vertex Cover},  such as those using the high-degree rule, the crown decomposition rule, or the Nemhauser-Trotter rule~\cite{Book-kernels}, are \lop-kernels. More involved kernels, such as those based on protrusion replacement~\cite{BodlaenderFLPST16}, are also \lop-kernels. However, we discussed in \autoref{sec:applications-lop} an example of a polynomial kernel for a vertex-minimization problem, namely \textsc{Tree Deletion Set}~\cite{GiannopoulouLSS16}, which is {\sl not} a \lop-kernel. We still do not know of a similar example that is a vertex-maximization problem.


For several technical reasons, we think that the framework of \lop-kernels seems to be more suited for maximization problems. In this direction, we showed that a direct application of our general result for vertex-maximization problems (\autoref{cor:lop-kernels}) yields kernelization lower bounds for {\sc MMVC} (\autoref{cor:lop1}) and {\sc MMFVS} (\autoref{cor:lop2}), matching the sizes of the best known kernels for these problems. We also presented consequences of our results for the \textsc{Maximum Independent Set} problem restricted to $K_t$-free graphs (\autoref{cor:lop4}) and conjectured (\autoref{conj:lop}) that, for every fixed graph $H$, the {\sc Maximum Independent Set} problem restricted to $H$-free graphs admits a polynomial \lop-kernel. For (vertex-)minimization problems, the only application that we were able to find concerns the \textsc{Tree Deletion Set} problem (\autoref{cor:lop3}).

We believe that our results could be applied to other vertex-maximization problems, in particular to the ``max-min'' version of other vertex-minimization problems, as they seem to be quite hard to approximate. It would  be interesting to find other examples of vertex-minimization problems, other than \textsc{Tree Deletion Set}, where our results  could be applied. Here, the natural candidates seem to be the ``min-max'' version of vertex-maximization problems, which seem to have been almost unexplored so far.

Our general results for maximization (\autoref{thm:lop-kernelsG}) and minimization (\autoref{thm:lop-kernelsGmin}) problems take into account an upper bound function $\ub(n) = \O(n^a)$ that upper-bounds the size of optimal solutions of the considered problems. All our applications presented in \autoref{sec:applications-lop} correspond to vertex problems, that is, to the case $a=1$. We leave for further research to find applications of our results for problems with superlinear upper bound functions, such as edge problems, for which $a=2$.

Our results are also able to derive lower bounds on the multiplicative constants of the existing kernels (cf. for instance the second item of \autoref{cor:lop-kernels} and \autoref{cor:lop-kernelsmin}). We still do not have any relevant application of this type for a concrete problem. For instance, if we apply \autoref{cor:lop-kernelsmin} to the \textsc{Vertex Cover} problem, relying on the non-existence of a $(2- \varepsilon)$-approximation under the Unique Games Conjecture~\cite{KhotR08}, we rule out the existence of a \lop-kernel with $(1-\varepsilon)k$ vertices, which is not particularly  interesting.

\medskip



%

We presented (\autoref{sec:subquadratic-kernels-MMVC}) subquadratic kernels for \textsc{Maximum Minimal Vertex Cover} on $H$-free graphs for some graphs $H$ satisfying the (constructive) Erd\H{o}s-Hajnal property, such as the bull, the complete graphs, or the paw. It would be interesting to obtain subquadratic kernels for other graphs $H$ satisfying the Erd\H{o}s-Hajnal property, such as $C_4$, the diamond, $P_5$, or $C_5$.  Note that, from~\cite{Gyarfas13}, $C_4$ and the diamond satisfy the constructive Erd\H{o}s-Hajnal property with constant $\delta \geq 1/3$. Note also that the graphs constructed in the reduction of \autoref{thm:no-poly-kernel} are $\{C_5,\text{diamond}\}$-free, as they are bipartite, hence
{\sc MMVC} is \NP-hard on this class, in contrast to the fact (\autoref{cor:diamond-P5}) that {\sc MMVC} can be solved in linear time on $\{P_5,\text{diamond}\}$-free graphs.  To the best of our knowledge, the complexity on $P_5$-free graphs is open, as well as on $K_{1,t}$ graphs for $t \geq 3$ (see \autoref{lem:kernel_K1t}). It is worth mentioning that $P_5$-free graphs have unbounded cliquewidth, because co-bipartite graphs, which are $P_5$-free, have unbounded cliquewidth.

As defined in \autoref{sec:prelim}, for a graph $G$  let ${\sf mmvc}(G)$ be the maximum size of a minimal vertex cover of $G$. Boria et al.~\cite{BoriaCP15} proved that if $G$ is an $n$-vertex graph without isolated vertices, then  ${\sf mmvc}(G) \geq \lfloor n^{1/2} \rfloor$. Note that this immediately yields a quadratic kernel for {\sc MMVC}: if $k \leq \lfloor n^{1/2} \rfloor$ we answer ``\yes'', otherwise $n  \leq k^2$. By the same argument, if ${\cal C}$ is a graph class such that every $n$-vertex graph $G \in {\cal C}$ without isolated vertices satisfies ${\sf mmvc}(G) \geq n^{1/2 + \varepsilon}$, for some $\varepsilon > 0$, then {\sc MMVC} restricted to ${\cal C}$ admits a (subquadratic) kernel with at most $k^{\frac{2}{1+2\varepsilon}}$ vertices. It might be possible that this is the case for some of the $H$-free graph classes for which we provided subquadratic kernels in \autoref{sec:subquadratic-kernels-MMVC}: we were not able to find any counterexample, that is, a family of $n$-vertex $H$-free graphs $G$ for which ${\sf mmvc}(G) = \Theta(n^{1/2})$. In particular, the case of triangle-free graphs seems particularly interesting. Haviland~\cite{Haviland08}
and Goddard and Lyle~\cite{GoddardL12} established upper bounds on the size of a minimum independent dominating set (that is, the complement of a minimal vertex cover) of triangle-free graphs. It follows from their results~\cite{Haviland08,GoddardL12} that there exist $n$-vertex triangle-free graphs $G$ with ${\sf mmvc}(G) = \Theta(n^{2/3} \cdot \log n)$, hence if such a constant $\varepsilon > 0 $ as discussed above exists for triangle-free graphs, necessarily $\varepsilon \leq \frac{2}{3} - \frac{1}{2} = \frac{1}{6}$. Therefore, the smallest kernel that we may obtain in this way on triangle-free graphs would have $k^{\frac{2}{1+2\varepsilon}} \leq k^{3/2}$ vertices, which matches the size of the kernel that we obtained in \autoref{thm:kernel-Kt-free} for the particular case $t=3$, disregarding lower-order terms and  multiplicative constants. Finding such a constant $\varepsilon > 0$ on $H$-free graphs for small graphs $H$, in particular on triangle-free graphs, looks like a challenging problem, having  interesting connections with the Ramsey numbers~\cite{Haviland08,GoddardL12}.

\medskip

\noindent \textbf{Acknowledgement}.  We would like to thank Michael Lampis (resp. Magnus Wahlstr{\"{o}}m, Venkatesh Raman) for pointing us to reference~\cite{DubloisHGLM20} (resp. reference~\cite{GiannopoulouLSS16}, references~\cite{BiswasRS20,KratschPRR14}).

\bibliography{Biblio-Ignasi}

\end{document}